\documentclass[12pt,reqno]{amsart}
\usepackage{amsmath,amsfonts,amsbsy,amsthm}
\usepackage{amssymb}
\usepackage{eucal}
\usepackage{dsfont}
\usepackage[normalem]{ulem}
\usepackage{comment}

\usepackage{xy}
\xyoption{matrix} \xyoption{arrow} \xyoption{arc} \xyoption{color}

\UseCrayolaColors
\usepackage{enumerate}
\usepackage[shortlabels]{enumitem}
\setlist{leftmargin=*}

\usepackage{tikz}
\usetikzlibrary{decorations.pathmorphing,shapes,arrows}

\newcounter{sarrow}

\numberwithin{equation}{section}

\makeatletter
\newtheoremstyle{corsivo}
   {\medskipamount}{\medskipamount}%
   {\itshape}{}%
   {\bfseries}{}%
   { }
   {\thmname{#1}\thmnumber{\@ifnotempty{#1}{ }\@upn{#2}}%
    \thmnote{ {\bfseries(#3)}}.}%
\makeatother

\theoremstyle{corsivo}
\newtheorem{thm}{Theorem}[section]
\newtheorem{lemma}[thm]{Lemma}
\newtheorem{crl}[thm]{Corollary}
\newtheorem{prop}[thm]{Proposition}
\newtheorem{conj}[thm]{Conjecture}
\newtheorem{dfn}[thm]{Definition}

\makeatletter
\newtheoremstyle{dritto}
   {\medskipamount}{\medskipamount}%
   {\rmfamily}{}%
   {\bfseries}{}%
   { }
   {\thmname{#1}\thmnumber{\@ifnotempty{#1}{ }\@upn{#2}}%
    \thmnote{ {\bfseries(#3)}}.}%
\makeatother

\theoremstyle{dritto}

\newtheorem{rmk}[thm]{Remark}
\newtheorem{example}[thm]{Example}

\newcommand{\sub}[1]{_{\mathrm{#1}}}
\newcommand{\abs}[1]{\left\lvert#1\right\rvert}

\newcommand{\ga}{\gamma}  
\newcommand{\la}{\lambda}

\newcommand{\Id}{\mathds{1}}  %%%%% Identity operator on H
\newcommand{\eu}{\mathrm{e}}
\newcommand{\iu}{\mathrm{i}}   %%%%% Imaginary unit

\newcommand{\B}{\mathbb{B}}                     %%%% Brillouin zone

\newcommand{\N}{\mathbb{N}}
\newcommand{\Z}{\mathbb{Z}}

\newcommand{\R}{\mathbb{R}}
\newcommand{\C}{\mathbb{C}}
\newcommand{\T}{\mathbb{T}}

\newcommand{\U}{\mathcal{U}}

\newcommand{\inner}[2]{\left\langle #1, #2 \right\rangle}

\newcommand{\norm}[1]{\left\| #1 \right\|}

\newcommand{\bra}[1]{\left\langle #1 \right|}
\newcommand{\ket}[1]{\left| #1 \right\rangle}

\newcommand{\set}[1]{ \left\{  #1 \right\}} 
\newcommand{\ve}[1]{\mathbf{#1}}

\DeclareMathOperator{\Tr}{Tr}         %  Hilbert space trace
\DeclareMathOperator{\Ran}{Ran} 
\DeclareMathOperator{\Span}{Span}

\DeclareMathOperator{\dist}{dist}
\DeclareMathOperator{\expo}{e}

\newcommand{\ie}{{\sl i.\,e.\ }}   %%% id est
\newcommand{\eg}{{\sl e.\,g.\ }} %%%  exemplum gratiae
\newcommand{\cf}{{\sl cf.\ }}      %%%  cum fero

\newcommand{\virg}[1]{``#1''}
\newcommand{\crucial}[1]{{\it \textbf{#1}}}
\newcommand{\half}{\mbox{\footnotesize $\frac{1}{2}$}}

\newcommand{\jp}[1]{\langle#1\rangle}

\newcommand{\w}{\psi}  % Generalised Wannier Functions
\newcommand{\Tuv}{\tau} % Trace per unit volume
\newcommand{\Loc}{G}
\newcommand{\mo}{\mathcal{I}_{\textsc{mo}}}
\newcommand{\mi}{\mathcal{I}_{\textsc{mi}}}

\newcommand{\NP}{\Gamma}
\newcommand{\Lattice}{\mathfrak{D}}

\newcommand{\kk}{{\bf k}}

\newcommand{\x}{{\bf x}}
\newcommand{\z}{{\bf z}}

\newcommand{\y}{{\bf y}}

\usepackage{color}

\newcommand{\E}{{\mathrm{e}}}

\newcommand{\Or}{{\mathcal{O}}}

\let\oldfootnote\footnote
\renewcommand{\footnote}[1]{\oldfootnote{\  #1}}

\title[Localization implies Chern triviality in non-periodic insulators]%%%
{Localization of generalized Wannier bases
\\[2mm] implies Chern triviality\\[2mm] in non-periodic insulators}
\author[G. Marcelli, M. Moscolari, G. Panati]{Giovanna Marcelli, Massimo Moscolari, Gianluca Panati}
\date{\today. Final version for arXiv. Paper published in Ann. Henri Poincar\'e, 24, 895-930 (2023). DOI: 10.1007/s00023-022-01232-7 .}

\usepackage{hyperref}

\begin{document}

\begin{abstract} We investigate the relation between the localization of generalized Wannier bases 
and the topological properties of two-dimensional gapped quantum systems of independent electrons in a disordered background, including magnetic fields, as in the case of Chern insulators and quantum Hall systems. We prove that the existence of a well-localized generalized Wannier basis for the Fermi projection implies the vanishing of the Chern character, which is proportional to the Hall conductivity in the linear response regime. 
Moreover, we state a localization dichotomy conjecture for general non-periodic gapped quantum systems.
\end{abstract}

\maketitle

\tableofcontents

\goodbreak

%%%%%  SECTION 1  %%%%%%%%%%%%%%%%%%%%%%%%%%%%%%%%

\section{Introduction}
\label{Sec:Intro}

Wannier bases have become, in the last few decades, a fundamental tool in theoretical and computational solid-state physics, as they provide a reasonable compromise between localization in position space and localization in energy, as far as compatible with the uncertainty principle \cite{WannierReview}.
Whenever it is well-localized, a Wannier basis: 
\renewcommand{\labelenumi}{{\rm(\roman{enumi})}} \vspace{-3mm}
\begin{enumerate} 
\item allows to implement numerical algorithms whose computational costs scale only linearly with the system size \cite{Goedecker};
\item provides a key tool for a simple and transparent description of macroscopic polarization and orbital magnetization in solids, yielding to computable formulae \cite{KSV,CeresoliThonhauserVanderbiltResta2006},  
later proved in a broader setting by more advanced mathematical techniques \cite{PanatiSparberTeufel2009, SchulzBaldesTeufel13}; 
\item  allows an efficient numerical treatment of deformed periodic systems \cite{ELu2011};    
\item  helps to justify the so-called \virg{atomic limit}, \ie the description of macroscopic solids as \virg{consisting of well-localized atoms}, a classical paradigm whose violation has opened the new field of Topological Chemistry \cite{BradlynEtAL}.    
\end{enumerate} 
\vspace{-3mm}
Last but not least, the variational characterization of Wannier functions proposed by Marzari and Vanderbilt, has turned Wannier bases into an efficient and flexible computational tool \cite{MaVa, WannierReview, PanatiPisante2013}.

Since these advantages rely on good decay properties of Wannier functions, several works have been dedicated to the rigorous proof of the existence of an exponentially localized Wannier basis in periodic (time-reversal symmetric) insulators.
These research efforts started with the pioneering work of Kohn and de Cloizeaux 
\cite{Kohn59, Cloizeaux1964, Cloizeaux1964a}, 
continued with the work of Nenciu and Helffer-Sj\"ostrand
\cite{Nenciu1983, Nenciu1991, HeSj89}, 
till the modern bundle-theoretic methods \cite{BrouderPanati2007, Panati2007}. 
More recently, the emphasis has shifted from abstract to \emph{algorithmic} proofs of existence \cite{FiMoPa_2, CoHeNe2015,CorneanMonaco17,CorneanMonacoTeufel17,CorneanMonacoMoscolari2019}, which allow a direct numerical implementation \cite{CaLePaSt2016}. 
Neglecting the linear independence condition, the related concept of \emph{Parseval frame} has also been investigated \cite{Kuchment2009,AucklyKuchment,CorneanMonacoMoscolari2019}.

When the Hamiltonian breaks time-reversal symmetry (TRS), as \eg in Chern insulators and Quantum Hall systems, the existence issue becomes more involved and interesting. 
A \emph{Localization-Topology Correspondence}, also dubbed \emph{Localization Dichotomy}, 
has been noticed and proved in \cite{MonacoPanatiPisanteTeufel2018}. There the authors proved that in a gapped $\Gamma$-periodic insulator in dimension $d \leq 3$, with $\Gamma \simeq \Z^d$ a Bravais lattice, a Wannier basis $ w =\set{w_{\ga, a}}_{\ga \in \Gamma, 1 \leq a \leq m}$
which is well-localized, in the sense that there exists $M_{*}<\infty$ such that
$$
\jp{X^2}_w  :=   \sum_{a=1}^m \int_{\R^d} |\x|^2 \, |w_{\ga, a}(\x + \ga)|^2  d\x \leq M_*   
\qquad \forall \ga \in \Gamma, 
$$
exists if and only if the Bloch bundle associated with the Fermi projection $P$  is Chern trivial. 
In $d=2$, the latter condition is equivalent to the vanishing of the (first) \emph{Chern number}, defined by 
\begin{equation} \label{Chern_number_2}
c_1(P) = 
\frac{1}{2\pi} \int_{\T^2_*}   \Tr \Big( 
P(\kk) \big[ \partial_{k_1}P(\kk) , \partial_{k_2}P(\kk) \big] 
\Big) \, dk_1 \wedge dk_2, 
\end{equation}
where $\T^2_* = \R^2/\Gamma^*$ is the $2$-dimensional Brillouin torus, 
and $ \U\sub{BF} \, P \, \U\sub{BF}^{-1}=  \int^{\oplus}_{\T^2_*} P(\kk) \, d\kk$ 
is the Bloch-Floquet-Zak decomposition of  $P$, see \eg \cite{Kuchment2016} or \cite{MonacoPanati2015}.  \newline
For $d=3$ the Chern triviality corresponds to the vanishing of three \virg{first Chern numbers} defined,  
for $i \neq j \in \set{1,2,3}$, by
\footnote{The vanishing of the numbers below expresses the vanishing of the \emph{first Chern class} $c_1(P)$
as an element of the cohomology space $H^1(\T^3_*, \R)$, 
which also implies the vanishing of the \emph{integer} first Chern class in $H^1(\T^3_*, \Z)$,
see \cite{Panati2007}. 
}\  
\begin{equation} \label{Chern_number_3}
c_1(P)_{ij} = \frac{1}{2\pi} \int_{\B_{ij}} 
 \Tr \Big( 
P(\kk) \big[ \partial_{k_i}P(\kk) , \partial_{k_j} P(\kk) \big] 
\Big) \, dk_i \wedge dk_j, 
\end{equation}
where $\B_{ij} \subset \T^3_*$ is the 2-dimensional subtorus of the Brillouin torus $\T^3_* = \R^3/\Gamma^*$ 
obtained by fixing the coordinate different from the $i$-th and the $j$-th (\eg equal to zero). 
Moreover, if the Bloch bundle is Chern trivial, then an exponentially localized Wannier basis always exists for $d\leq 3$ 
\cite{BrouderPanati2007, Panati2007}.  \newline
To avoid any source of confusion, we emphasize here that the  Localization-Topology Correspondence (LTC) 
refers to a different kind of localization mechanism than the one appearing in the well-known Anderson localization. 
Moreover, while Anderson localization concerns the decay properties of \emph{eigenfunctions} 
of Schr\"odinger operators with random potentials, the LTC refers to the decay properties of particular orthornormal bases, namely generalized Wannier bases, that span the spectral subspaces of a (possibly non-periodic) Hamiltonian operator.

This paper aims at the generalization of the LTC from the periodic setting considered in \cite{MonacoPanatiPisanteTeufel2018} to the non-periodic one.  Since both sides of the correspondence, namely Wannier bases and Chern numbers, are defined by using periodicity in an essential way, even the formulation of  a reasonable conjecture requires some care. 

On the side of Wannier bases, a generalization of this concept to non-periodic systems has been discussed since the early work of Kohn and Onffroy \cite{KohnOnffroy1973}. Later, Kivelson noticed that for $d=1$ a generalized Wannier basis is provided by the eigenfunctions of the reduced position operator $\widetilde X := PXP$, where $X$ is the usual position operator \cite{Kivelson1982}. This intuition has been put on solid mathematical grounds in \cite{NenciuNenciu1998}, where it is proved that for any gapped 1-dimensional Schr\"odinger operator the spectrum of $\widetilde X $ is discrete, and the corresponding eigenfunctions form a \crucial{generalized Wannier basis} (GWB)%%
%%%%% FOOTNOTE %%%%%%%%%%%%%%%%%%%%
\footnote{Whenever we use the adjective \virg{generalized} we refer to Wannier functions for non-periodic systems. 
However, this terminology is far from universal. The reader is warned that the adjective \virg{generalized} has been occasionally used referring to Wannier functions for a multi-band periodic system (which are called \emph{composite} Wannier functions in the mathematical literature), as \eg in \cite{MaVa}.     
}, %%%%%%%%%%%%%
as reviewed in Example \ref{exampe:PXP}.  
While the above construction does not generalize to $d>1$,  the existence of a GWB can be proved for several 
specific $d$-dimensional systems, as discussed in Section \ref{Sec:GWB} following \cite{NenciuNenciu1993}
and \cite{DeMonvelNenciuNenciu1995}. 
Moreover, in \cite{CorneanNenciuNenciu2008} it has been shown that the construction of \cite{NenciuNenciu1998} is optimal, in the sense that the obtained GWB has the same exponential decay of the associated Fermi projection. 
An alternative strategy to construct a GWB for a generic $2$-dimensional non-periodic system has been recently suggested, and numerically validated \cite{StubbsWatsonLu20}. We shortly review the whole topic of GWB in Section~\ref{Sec:GWB}.

 %%%%%%%%%%%%%%%%%%%%%%%%%%%%

On the other side of the correspondence, for $d=2$ the Chern number of the Bloch bundle is naturally generalized to non-periodic models  by the  \crucial{Chern character}, defined,  for a sufficiently regular orthogonal projection $P$ acting on $L^2(\R^2)$, by 
\begin{equation} \label{Chern_marker_intro}
C(P)  := 2\pi  \,\, \Tuv  \Big( \iu P  \big[ [X_1, P], [X_2, P] \big] \Big)
\end{equation}
where $\Tuv(\cdot)$ is the trace per unit volume. Whenever $P$ is periodic the above formula reduces to 
\eqref{Chern_number_2}, so that $C(P)=c_1(P)$.  Physically, $\frac{1}{2\pi}C(P)$  gives in Hartree units%%
%%
%%%%%FOOTNOTE %%%%%%%%%%%%%%%%%%%%%%%%%%%%
\footnote{In this system of units the reduced Planck constant, the electron charge and the electron mass are dimensionless and equal to $1$. In particular, the quantum of charge conductivity in the quantum Hall effect is $\frac{e^2}{h}=\frac{1}{2\pi}$, where $-e$ is the electron charge and $h$ is the Planck constant. %%%%%%  
}   the Hall conductivity of the system,  which agrees with the Hall conductance under mild technical assumptions 
\cite{AvronSeilerSimon1994}.

In connection with solid state physics, formula \eqref{Chern_marker_intro} first appeared -- to the best of our knowledge -- 
in 1986  in a conference proceedings by J.~Bellissard \cite{Bellissard1986}, where $C(P)$ is baptized Chern character and it is specified that formula \eqref{Chern_marker_intro} applies to orthogonal projectors affiliated to a specific $C^*$-algebra of \emph{ergodic} operators (reducing to \emph{periodic} operators in the deterministic case). However, as early envisaged by Bellissard himself, the same formula makes sense in a broader context, for projectors $P$  whose kernel is sufficiently fast decreasing away from the diagonal   \cite{NakamuraBellissard}. In an ergodic setting, this viewpoint and its relation with non-commutative geometry have been deeply explored in \cite{BellissardElstSchulz-Baldes1994}.
The relation with the index of a pair of projections and with the Fredholm index has been also clarified \cite{Bellissard1986, AvronSeilerSimon1994, KellendonkSchulzBaldes19}, see the review paper \cite{Graf2007} and references therein. \newline
The same formula \eqref{Chern_marker_intro} has later been reconsidered in a non-ergodic and non-covariant setting, 
provided the trace per unit volume exists, which it happens in particular for exponentially localized projections, 
as in Definition \ref{dfn:ELP}.  
Interpreted in this broader sense, formula \eqref{Chern_marker_intro} still produces an integer whenever the trace per unit volume exists, as proved in \cite{ElgartGrafSchenker} for discrete models.
The latter proof, which is essentially based on the identity between the Chern character and the index of a
Fredholm operator as previously established in the covariant and ergodic setting \cite{AvronSeilerSimon1994,KellendonkSchulzBaldes19}, 
is generalized to gapped continuum models in Proposition \ref{prop:CIndex}. \newline    
Finally, in 2011 physicists rediscovered an equivalent version of formula  \eqref{Chern_marker_intro}, 
namely  $C(P) = 2\pi \iu \, \tau( [ P X_1 P,  P X_2 P])$, labelling the corresponding quantity as \emph{Chern invariant} \cite{BiancoResta2011} and \emph{Chern marker} \cite{Caio et al 2019, Irsigler et al 2019}.
The latter name has been also used in the recent mathematical literature \cite{MarcelliMonacoMoscolariPanati2018, CorneanMonacoMoscolari2018}.

%%%%%%%  In this paper  %%%%%%%%%%%%%%%%%%%%%%%%%%%%%%%%%%

In this paper, we conjecture a Localization-Topology Correspondence for non-periodic gapped systems 
(Conjecture \ref{ConjectureLocDic}) and we prove one of the conjectured implications, with a non-optimal
threshold (Theorem \ref{MainTheorem}):  
if an orthogonal projection $P$, which acts on $L^2(\R^2)$ and is exponentially localized 
in the sense of Definition~\ref{dfn:ELP}, 
admits a GWB  $\{\w_{\gamma,a}\}_{\gamma \in \Lattice, 1 \leq a \leq m(\gamma)}$ 
(here $\Lattice$ is a discrete set, as in Definition~\ref{GWB})
which is $s$-localized for some $s>4$, 
\ie for $s>4$ there exists $M<+\infty$ such that
$$
\int_{\R^2} \jp{\x - \ga}^{2s} \,\, |\w_{\gamma,a}(\x)|^2 \,d\x \leq M 
\qquad \forall\gamma \in \Lattice,\,1 \leq   a \leq  m(\gamma),
$$
\noindent then the corresponding Chern character $C(P)$ vanishes. Notice that neither periodicity nor covariance with respect to the action of a group is required. The fact that our result does not assume periodicity or covariance makes it suitable 
for applications also to random Schr\"odinger operators \cite{AizenmanWarzel}. 
In that context, it has been proved that localized eigenfunctions, whose localization might be caused by very different mechanisms, do not contribute to the conductivity: the case of localization of eigenfunctions due to deep wells has been analyzed in \cite{NakamuraBellissard}; while the case of Anderson localization has been treated in \cite{GerminetKleinSchenker2007}, see also the earlier work \cite{Kunz}. 
Although Wannier functions are not eigenfunctions, our results say
-- coherently with the latter papers --  that the existence of a well-localized GWB implies the vanishing of the 
transverse charge conductivity.   

%%%%%%   Further references %%%%%%%%%%%%%%%%%%%%%%%%

\subsection{Further references}
We conclude the Introduction mentioning some results appeared during the revision of our manuscript. The result in this paper was first announced in \cite{MarcelliMonacoMoscolariPanati2018} and a preliminary version of the proof was provided in the PhD thesis of one of the Authors \cite{Moscolari2018}. 
These preliminary papers, together with the results in \cite{MonacoPanatiPisanteTeufel2018}, resparked the interest of part of the community for the analysis of Wannier bases for non-periodic systems. Besides the aforementioned \cite{StubbsWatsonLu20,StubbsWatsonLu20II}, we notice the preprint by Lu and Stubbs \cite{LuStubbs21II} (see also \cite{LuStubbs21I}) where they manage to show Theorem \ref{MainTheorem} with $s>1$. 
Nevertheless, Theorem \ref{MainTheorem} with the optimal threshold $s\geq 1$ is still an open problem.  Finally, it is worth to notice that the LTC has been recently generalized also in a different direction, within the $C^*$-algebraic approach to solid state physics \cite{LudewigThiang19,BourneMesland}.

%%%%%%%%%%%%%%%%%%%%%%%%%%%%%%%%%%%%%%%%%%%%%%

\medskip

\noindent \textbf{Acknowledgements.}  We are grateful to H.~Cornean, T.~Loring, D.~Monaco, G.~Nenciu, 
M.~Porta,  G.\,C.~Thiang, and S.~Teufel for several stimulating discussions on related aspects of the theory of Wannier bases and of  transport theory. We thank the anonymous Reviewers for several stimulating comments and remarks on the previous version of the paper. \newline
G.\,M. gratefully acknowledges the financial support from the European Research Council (ERC), under the European Union's Horizon 2020 research and innovation programme (ERC Starting Grant MaMBoQ, no. 802901). M.\,M. gratefully acknowledges the financial support from Grant 8021-00084B of the Danish Council for Independent Research $|$ Natural Sciences.
The work of M.\,M. is supported by a fellowship of the Alexander von Humboldt Foundation.

%%%%%%  SECTION 2  %%%%%%%%%%%%%%%%%%%%%%%%%%%%%%%%%%%
\goodbreak
\section{Setting and fundamental concepts} \label{sec:settings}

As explained in the Introduction, the aim of this paper is to generalize results from the periodic setting to the non-periodic one. In order to model materials that are not exactly crystalline we have to replace the {\it Bravais lattice}%%%
\footnote{Recall that a Bravais lattice $\Gamma \subset \R^d$ is defined as a discrete subgroup of $(\R^d, +)$ with maximal rank. It follows that there exist a (non-unique) linear basis $\set{\ve a_1, \ldots, \ve a_d} \subset{\R^d}$ such that 
	$\Gamma = \Span_{\Z}\set{\ve a_1, \ldots, \ve a_d} \simeq \Z^d$. 
},\  
which models the periodicity of a crystalline system, by a discrete set $\Lattice$.  We require some uniformity as in the following definition, where $B_\rho(\x)\subset \R^d$ denotes the open ball of radius $\rho>0$ centered in $\x\in \R^d$.
\begin{dfn}
A set $\Lattice \subset \R^d$ is said to be
\emph{$r$-uniformly discrete} if  there exists $r>0$ such that $\forall \,\x \in \R^d$ the set $B_{r}(\x) \cap \Lattice$ contains at most one element.
\end{dfn} 

\noindent Obviously, a Bravais lattice $\Gamma\simeq\Z^d$ is a $r$-uniformly discrete set for a suitable $r>0$.

A central role in our analysis will be played by the concept of localization in space, hence the following Definition will be useful.
\begin{dfn}[Localization function]
We say that a continuous function \[G\colon[0,\infty)\to(0,\infty)\] is a \emph{localization function} if $\lim_{x \to\infty}G(x)=+\infty$ and there exists a constant $C_G>0$ such that
\begin{equation}
\label{eqn:GTriang}
G(\norm{\x-\y})\leq C_G \,G(\norm{\x-\z})G(\norm{\z-\y})\qquad \forall \,\x,\y,\z\in\R^d.
\end{equation}
\end{dfn}

\noindent Natural examples of localization functions are $\Loc(x) = \expo^{\alpha x}$ for some $\alpha > 0$, 
and  $\Loc(x)=\langle x \rangle^{2s}$ for some $s > 0$, where  $\langle x \rangle:= \left(1+ x^2 \right)^{\frac{1}{2}}$ as usual.

\medskip

Our aim is to investigate the relation between localization of GWBs and transport properties in non-interacting gapped quantum systems, whose dynamics is generated by a one-particle (magnetic) Schr\"odinger operator \cite{AvronHerbstSimon1981}. 
The spectral projections onto an isolated component of the spectrum of such operators provide archetypal examples of exponentially localized projections, 
defined as follows. 

\begin{dfn}[Exponentially localized projection]
	\label{dfn:ELP}
	We say that an orthogonal projection $P$ acting on $L^2(\R^d)$ is \emph{exponentially localized} if $P$ 
	is an integral  operator with a \emph{jointly continuous} integral kernel $P(\cdot\,,\,\cdot)\colon\R^d\times\R^d \to \C$ and 
	there exist two constants $C,\beta >0$ such that
	\begin{equation}
	\label{DiagLocIntKErnel}
	\left|P(\x,\y)\right|\leq C \eu^{-\beta \|\x-\y\|}\,\qquad \forall \, \x,\y \in\R^d.
	\end{equation}
\end{dfn}

The following Proposition provides examples of exponentially localized projection for a large class of $2$-dimensional quantum systems, including magnetic Schr\"odinger operators with constant magnetic field. 
\begin{prop}
\label{prop:GQS}
Let $V: \R^2 \to \R$ be in $L^2_{\rm uloc}(\R^2)$, which means that $V$ is uniformly locally square-integrable, \ie
\begin{equation}
\label{UniformL2loc}
\sup_{\x \in \R^2} \int\limits_{\|\x-\y\|\leq 1} |V (\y)|^2  d\y \, < \, \infty.
\end{equation}

\noindent Assume that the magnetic vector potential ${\bf A}: \R^2 \to \R^2$ is in $L^4_{{\rm loc}}(\R^2,\R^2)$ with distributional derivative
$\nabla \cdot {\bf A} \in L^2_{{\rm loc}}(\R^2)$. Consider the Hamiltonian operator
\begin{equation*}
H_{\bf A}:=- \half \Delta_{\bf A} + V \, , 
\end{equation*} 
where $-\Delta_{\bf A}:=\left(-\iu\nabla - \mathbf{A}\right)^2$. Then
\begin{enumerate} [label=(\roman*),ref=(\roman*)]
	\item \label{item:C1} $H_{\bf A}$ is essentially selfadjoint on $C^\infty_{\rm c}(\R^2)$. We denote its closure again by $H_{\bf A}$.
	\item \label{item:C2} The domain $\mathcal{D}_{\bf A}$ of selfadjointness of $H_{\bf A}$ is contained in $C(\R^2)$.
	\item \label{BoundedFromBelow} $H_{\bf A}$ is bounded from below.
\end{enumerate}
Moreover, assume that $\sigma(H_{\bf A})$,  the spectrum of $H_{\bf A}$, has an {\emph{isolated component}}, 
\ie there exist two non empty sets $\sigma_0,\sigma_1 \subset \R$ and $E_{\pm} \in \R\setminus\sigma(H_{\bf A})$ such that
\begin{equation}\label{gap}
\sigma(H_{\bf A})=\sigma_0 \cup \sigma_1,  \qquad \sigma_0 \subset \left(E_-,E_+\right) \: 
\text{ and }\: \left(E_-,E_+\right) \cap \sigma(H_{\bf A})= \sigma_0.
\end{equation}
Let $P_0$ be the spectral projection corresponding to $\sigma_0$. Then,
\begin{enumerate}[label=(\roman*),ref=(\roman*), resume]
 \item \label{item:FermiEL} $P_0$ is an exponentially localized projection in the sense of Definition~\ref{dfn:ELP}.
\end{enumerate}
\end{prop}
\noindent We refer to $P_0$ as the \emph{Fermi projection}, interpreting it as the projection onto the space of \virg{occupied} states of a system of non-interacting particles.

\begin{proof}[Sketch of the proof]
The statements \ref{item:C1} and \ref{BoundedFromBelow} are a special case of \cite[Theorem~3]{LeinfelderSimader1981}. The property \ref{item:C2} is a consequence of the fact that the Schr\"odinger semigroup $e^{-t H_A}$ is a bounded operator from $L^2(\R^2)$ to $L^\infty(\R^2)$ \cite[equation (2.40)]{BroderixHundertmarkLeschke2000} and maps $L^2(\R^2)$ into the space of continuous functions \cite[Theorem 4.1]{BroderixHundertmarkLeschke2000}. Therefore, repeating the proof of \cite[Corollary B.3.2, Theorem B.3.3]{Simon1982} one can show that, for $\lambda>0$ large enough, the resolvent $\left(H_{\bf A} + \lambda \right)^{-1}$ maps $L^2(\R^2)$ into the space of continuous functions.  
Regarding the statement \ref{item:FermiEL}, the existence and joint continuity of the integral kernel of $P_0$ is a standard result  in the theory of Schr\"odinger operators \cite[Remarks 6.2.(ii)]{BroderixHundertmarkLeschke2000}\cite{Simon1982}.
The exponential localization of the integral kernel is a consequence of the Combes--Thomas estimates on the resolvent \cite{CT73}, which can be shown by adapting the proofs of \cite[Proposition 3.1, Appendix A]{CorneanNenciu09}, coupled with the spectral gap assumption on $\sigma_0$, which allows to choose an integration contour $\mathcal{C}\subset \C$ with a uniform positive distance from $\sigma_0$, and the fact that $P_0 =-\frac{\iu}{2 \pi} \oint_{\mathcal{C}} dz \, z (H_{\bf A}-z)^{-2}$, see \cite[Appendix A.1.1]{Moscolari2018} for details.
\end{proof}

%%%%%%%%%%%%%%%%%%%%%%%%%%%%%%%%%%%%%%%%%

\subsection{Generalized Wannier Bases: a short review}
\label{Sec:GWB}

Following the seminal idea in \cite{NenciuNenciu1993,NenciuNenciu1998} we define generalized Wannier bases and functions as follows. 

\begin{dfn}[Generalized Wannier basis] 
\label{GWB}
Let $P$ be an orthogonal projection acting on $L^2(\R^d)$ . We say that $P$ admits a \emph{generalized Wannier basis} (GWB) if there exist a localization function $\Loc$, a discrete set $\Lattice \subset \R^d$, a  constant $m_*>0$ and a set $\{\w_{\gamma,a}\}_{\gamma \in \Lattice, 1 \leq a \leq m(\gamma)}\subset L^2(\R^d)$ with $m(\gamma)\leq m_*$ for every $\gamma \in \Lattice$, such that:
\begin{enumerate}[label=(\roman*),ref=(\roman*)]
    \item $\{\w_{\gamma,a}\}_{\gamma \in \Lattice, 1 \leq a \leq m(\gamma)}$ is an orthonormal basis
 for $\Ran P$;
    \item there exists $M<\infty$ such that every $\w_{\gamma,a}$ is \emph{$\Loc$-localized around $\gamma$}, \ie
            	\begin{equation}
            	\label{GWFLoc}
            	\int_{\R^d} |\w_{\gamma,a}(\x)|^2 \Loc(\|\x-\gamma\|) \,  d\x \leq M \quad \forall\gamma \in \Lattice,\,1 \leq        
            		a \leq  m(\gamma).
            	\end{equation}
\end{enumerate}
If the above conditions are satisfied, we say that $P$ admits a GWB which is \emph{$\Loc$-localized around the set $\Lattice$}. One also says that  $\w_{\gamma,a}$ is a \emph{generalized Wannier function} (GWF) $\Loc$-localized around
$\gamma\in\Lattice$.    
\end{dfn}

Notice that the localization center $\gamma\in\Lattice$ is far from being unique, even in the periodic case. In the physics literature, the vector $\langle \w_{\gamma,a}, \mathbf{X} \w_{\gamma,a}\rangle$ 
is called \emph{centroid} of $\w_{\gamma,a}$. Moreover, in contrast with the usual definition of Wannier functions for periodic systems, the index $m(\gamma)$ might depend on $\gamma\in \mathfrak D$. Indeed, without a periodicity assumption, one may in general expect a different number of orbitals in different lattice sites. The only constraint we require is that such a number is uniformly bounded, namely $m(\gamma) \leq m_*, \,  \forall \gamma \in \mathfrak D$. 

Finally, it might seem natural to require, in the definition of GWB, that the set $\Lattice$ is a \emph{Delone set}, 
\ie both {uniformly discrete} and {uniformly nowhere sparse}%%%% Footnote
\footnote{\label{fn:nosparse} A discrete set is said to be \emph{$R$-uniformly nowhere sparse} if there exists $R>0$ such that $\forall\, \x \in \R^d$ the set $B_{R}(\x) \cap \Lattice$ contains at least one element.}. %%%%End 
However, as we do not use both these properties in our Theorem (only uniform discreteness is assumed), we preferred not to include them in the definition.

The following terminology agrees with \cite{MonacoPanatiPisanteTeufel2018}:
\begin{itemize}
\item[$\diamond$] if  \eqref{GWFLoc} holds true with $\Loc(\|\x\|)=\E^{2 \alpha \|\x\|}$, for some $\alpha > 0$,   
         we say that the GWB is \emph{exponentially localized};
\item[$\diamond$] if  \eqref{GWFLoc} holds true with $\Loc(\|\x\|)=\langle \x \rangle^{2s}$ for some $s > 0$,
we say that the GWB is \emph{s-localized}. 
\end{itemize}

In the case of an exponentially localized projection $P$, the $L^2$-localization estimate \eqref{GWFLoc} implies a $L^\infty$-estimate on any GWF $\w_{\gamma,{a}}$, as in the following
\begin{lemma}[From $L^2$-estimates to $L^\infty$-estimates]
	\label{Rem:bounds}
	Let $P$ be an exponentially localized projection acting on $L^2(\R^d)$.  
		Assume that $P$ admits a GWB $\{\w_{\gamma,a}\}_{\gamma \in \Lattice, 1 \leq a \leq m(\gamma)}$ 
		with localization function  $G(\|\x\|) \leq C_1 \eu^{\lambda \|\x\|}$, with $C_1 > 0$ and $\lambda < 2\beta$ for $\beta$  as in \eqref{DiagLocIntKErnel}. 
		Then, there exists a constant $K>0$, independent of $\gamma \in \Lattice$, such that each GWF $\w_{\gamma,{a}}$ satisfies
	\begin{equation}
	\label{LinftyEstimate}
	\left|\w_{\gamma,{a}} (\x) \right| \leq K \,  G(\|\x-\gamma\|)^{-1/2}   
	\quad \forall \, \x \in \R^d, \,\, \forall \, \gamma \in \Lattice.
	\end{equation}
\end{lemma}
\begin{proof}
	Using the fact that $\w_{\gamma,{a}}=P\w_{\gamma,{a}}$, one has
	\begin{equation*}
	\begin{aligned}
	& \left| G(\|\x-\gamma\|)^{1/2} \w_{\gamma,{a}} (\x) \right| \leq \int_{\R^d} d \y \, G(\|\x-\gamma\|)^{1/2}  \,  \left| P(\x,\y) \right| \left| \w_{\gamma,{a}}(\y)\right|
	\\
	&\leq C_G^{1/2}\left( \int_{\R^d} d \y \, G(\|\x-\y\|)  \,  \left| P(\x,\y) \right|^2 \right)^{1/2} \left( \int_{\R^d} d\y \,G(\|\gamma-\y\|) \,\left| \w_{\gamma,{a}}(\y)\right|^2 \right)^{1/2},
	\end{aligned} 
	\end{equation*}
	where in the second and third inequality we have used property \eqref{eqn:GTriang} and the Cauchy--Schwarz inequality, respectively.
	In view of \eqref{DiagLocIntKErnel} and of the hypothesis on $G$, there exists a constant $K>0$,
	independent of $\gamma \in \Lattice$, such that
	$$
	C^{1/2}_G \sup_{\x \in \R^d} \left( \int_{\R^d} d \y \, G(\|\x-\y\|)  \,  \left| P(\x,\y) \right|^2 \right)^{1/2} \left( \int_{\R^d} d \y \, G(\|\gamma-\y\|) \,\left| \w_{\gamma,{a}}(\y)\right|^2 \right)^{1/2}  \leq K\,.
	$$
	Therefore, \eqref{LinftyEstimate} is proved.
\end{proof}

\medskip 
Whenever the projection $P$ admits a GWB, it can be written as  (using the Dirac  notation)
\begin{equation}
\label{braket}
P\;=\; \sum_{\gamma \in \Lattice} \sum_{1\leq a \leq m(\ga)} \ket{\w_{\gamma,a}} \bra{\w_{\gamma,a}}
\end{equation}
where the series is understood as the strong limit of the finite sums of the projections on the one-dimensional spaces spanned by each GWF. Notice that the strong limit is independent of the ordering of the series. Moreover, if the set $\Lattice$ is $r$-uniformly discrete, it is convenient to rearrange the above series as  
$$
P \;\varphi=\;\lim_{L\to\infty} \sum_{\gamma \in \Lattice\cap\Lambda_L} \sum_{1\leq a \leq m(\ga)}  \bra{\w_{\gamma,a}}\varphi\rangle \w_{\gamma,a},\quad\forall\,\varphi\in L^2(\R^d),
$$
with $\Lambda_L = [-L,L]^d$. 
Furthermore, if $P= P_0$ is the Fermi projection of an operator satisfying the hypotheses 
of  Proposition~\ref{prop:GQS} and admitting a GWB that is exponentially localized or $s$-localized, with $s>1$, the equality 
\begin{equation} \label{braket2}
P_0(\x,\y)\;=\; \sum_{\gamma \in \Lattice} \sum_{1\leq a \leq m(\ga)} \w_{\gamma,a}(\x) \overline{\w_{\gamma,a}(\y)}\, , \qquad \forall\, \x,\y \in \R^2 ,
\end{equation}
holds true pointwise, since $\w_{\gamma,a}\in \Ran P_0\subset \mathcal{D}_{\bf A}\subset C(\R^2)$ by Proposition~\ref{prop:GQS}\ref{item:C2} and $\w_{\gamma,{a}}$ satisfies \eqref{LinftyEstimate}.

We emphasize a crucial point: the fact that $P_0$ is an exponentially localized projection does \emph{not} imply that the GWFs appearing in \eqref{braket2} are themselves exponentially localized, as it is explained in 
\cite{MonacoPanatiPisanteTeufel2018}   for the periodic case. While the exponential localization of $P_0$ is a 
mere consequence of the gap condition, via Combes--Thomas estimates, the existence of an exponentially localized 
Wannier basis is related - in the periodic case - to the Chern triviality of the vector bundle associated with $P_0$ \cite{Panati2007,BrouderPanati2007,FiMoPa_2,CoHeNe2015,Kuchment2016,CorneanMonacoMoscolari2019}. 
Notice, moreover, that not every orthonormal basis is able to capture the topological properties of the former vector bundle. Indeed, if one relaxes the definition of GWB, allowing \eg an orthonormal basis whose elements are localized around
circles of increasing radii as in \cite{Prodan}, it is always possible to find an exponentially localized orthonormal basis such that \eqref{braket2} holds true, see \cite{MoscolariPanati}. In this sense, the definition of GWB is a good compromise 
between generality and ability to \virg{read} the vanishing of the Chern character.

\goodbreak

Whenever the Hamiltonian operator is $\Gamma$-periodic and satisfies suitable regularity conditions, see \eg \cite[Remark 3.2]{MonacoPanatiPisanteTeufel2018}, the Fermi projection $P_0$ commutes with a unitary representation of the group $\Gamma \simeq \Z^d$, denoted by  $\left\{T_\gamma\right\}_{\ga \in \Gamma}$, 
and one constructs via Bloch-Floquet theory  the usual {\it composite Wannier basis} for $P_0$, denoted by $\{T_\gamma w_{0,a} \}_{\gamma \in \Gamma, 1 \leq a \leq m}$, for some fixed integer $m \geq 1$ (equal to the number of Bloch bands associated with $P_0$). 
It is easy to see that such composite Wannier basis  satisfies Definition~\ref{GWB}, so it is indeed a GWB. 
According to the value of the first Chern number(s) of $P$, such a GWB can be chosen 
exponentially localized if $c_1(P) = 0$,  or $s$-localized for any $s <1$ (but not for $s=1$!) if instead
$c_1(P) \neq 0$ \cite{MonacoPanatiPisanteTeufel2018}.

It is of relevance to note that there are non-periodic systems in which is not possible to construct a composite Wannier basis - in view of the lack of periodicity - but is still possible to have a generalized Wannier basis. In this direction, 
we review here some examples%%%
\footnote{See \eqref{UniformL2loc} for the definition of $L^2_{\rm{uloc}}(\R^2)$, which can be easily generalized to any dimension.}:

\begin{example}[Generic 1-dimensional insulators]
\label{exampe:PXP}
Consider the one-dimensional Schr\"odinger operator
\begin{equation}
\label{eqn:1DH}
H=-\frac{\mathrm d^2\phantom{x}}{{\mathrm d} x^2} + V \qquad \text{ with } V \in L^2_{\rm uloc} (\R) .
\end{equation}
$H$ is a selfadjoint operator in $L^2(\R)$, bounded from below. Assume that the spectrum of $H$ has an isolated component 
$\sigma_0$, such that the range of the spectral projection $P_0$ corresponding to $\sigma_0$ 
has infinite dimension. Inspired by Kivelson \cite{Kivelson1982}, the authors of  \cite{NenciuNenciu1998} consider the reduced position operator $P_0XP_0$ with domain 
$\mathcal{D}(P_0XP_0)=\mathcal{D}(X)\cap\Ran P_0$ and prove that:
\begin{enumerate}[label=(\roman*),ref=(\roman*)]
	\item the resolvent of $P_0XP_0$ is a compact operator, hence the spectrum of $P_0XP_0$ is purely discrete.
	The corresponding eigenvectors $\psi_{\gamma,a} \in L^2(\R)$ satisfy
	$$
	P_0XP_0 \psi_{\gamma,a}=\gamma \psi_{\gamma,a} \qquad \forall\,\gamma \in \sigma(P_0XP_0) ,\,1 \leq a \leq m(\gamma),
	$$
	where $m(\gamma)$ is the degeneracy index of the eigenvalue $\gamma$;
	\item there exist two constants $0<\alpha, M <\infty$ such that 
	$$
	\int_{\R} |\w_{\gamma,a}(\x)|^2 \eu^{ \alpha \|\x-\gamma\|} \,  d\x \leq M \quad \forall\,\gamma \in \sigma(P_0XP_0) ,\,1 \leq a \leq m(\gamma). 
	$$
\end{enumerate}

Setting  $\Lattice_0 = \sigma(P_0XP_0)$, the orthonormal basis 
$\left\{\w_{\gamma,{a}}\right\}_{\gamma \in \Lattice_0, 1 \leq a \leq m(\gamma) }$
is a GWB \emph{in the sense of Nenciu--Nenciu},  as defined in \cite{NenciuNenciu1998}, 
exponentially localized around the discrete set  $\Lattice_0$.
To prove that  $\left\{\w_{\gamma,{a}}\right\}$ is a GWB according to our (stronger) Definition \ref{GWB}, 
one has to show that $\exists \, m_* : m(\gamma)\leq m_*$ for all $\gamma \in \Lattice$, which is not known in general.
Whenever the Hamiltonian is periodic, the generalized Wannier basis constructed with the previous strategy \cite{NenciuNenciu1998} coincides with a usual composite Wannier basis \cite{Costa2014, Moscolari2018}.
Furthermore, in \cite{CorneanNenciuNenciu2008} it has been proved that the number of eigenvalues of $P_0XP_0$, counted with their multiplicity, contained in an interval of length $L$, grows at most linearly with $L$. However, this result does not imply that the set $\sigma(P_0XP_0)$ is $r$-uniformly discrete or $R$-uniformly nowhere sparse
 for some $r,R>0$.  The proof of the latter claims is, to our knowledge, still an open problem.
\end{example}

The previous construction for $1$-dimensional systems does not directly generalize to higher dimensions, essentially 
for two reasons: 
\renewcommand{\labelenumi}{{\rm(\roman{enumi})}}
\begin{enumerate} 
\item  generically, the compact resolvent property of $P_0XP_0$ holds true only in $L^2(\R)$, as one can easily see for $d=2$ when $P_0$ is $\Z^2$-periodic;
\item in general $P_0 X_i P_0$ and $P_0 X_jP_0$ do not commute for $i \neq j$, so the attempt to simultaneously diagonalize them fails. It is worth noticing that the topological and transport properties of the system are encoded exactly in the commutator between $P_0 X_1 P_0$ and $P_0 X_2 P_0$, see Section \ref{subsec:ChernMarker} and specifically equation \eqref{Chern_commutator}. The relation with the theory of almost-commuting operators has been explored in  \cite{HastingsLoring10}.
\end{enumerate} 
Some proposals to circumvent these difficulties appeared recently in \cite{StubbsWatsonLu20,StubbsWatsonLu20II}, where the authors show that, under the additional crucial assumptions of uniform spectral gaps for the spectrum of the operator $PX_1P$, it is possible to prove the existence of an exponentially decaying GWB for the projection $P$. Proving that $PX_1P$ satisfies such uniform spectral gaps hypothesis is still an open problem, but in \cite{StubbsWatsonLu20} the authors present numerical simulations showing that in explicit tight-binding models the spectrum of $PX_1P$ have spectral gaps only when the projection $P$ is Chern trivial, which is in accordance to Conjecture \ref{ConjectureLocDic}.

Despite the difficulties to deal with \emph{generic} non-periodic systems for $d>1$, 
there are several specific $d$-dimensional models for which the existence of a well-localized GWB has been proved, as in the following examples.

\begin{example}[$d$-dimensional insulators with weak disorder]
\label{example:Impurties}$ $\newline 
Generalized Wannier bases can be useful in the analysis of periodic systems perturbed by impurities and weak disorder. Given $\lambda \in \R$ and a Bravais lattice $\Gamma \subset \R^d$, $d\leq 3$, we  
consider the Hamiltonian operator, acting in $L^2(\R^d)$,
$$
H_{\lambda}=-\Delta + V_{\Gamma} + \lambda W
$$
where $V_\Gamma \in L^2_{\rm uloc}(\R^d)$ is $\Gamma$-periodic. Assume that $W\in L^2_{\rm uloc}(\R^d)$, so that $H_\lambda$ is an entire family of type $A$ in the sense of Kato.  The potential $V_\Gamma$ models the periodic background of the crystalline insulator, while $W$ models either a sum of localized impurities or a delocalized disorder.  \\
Assume that the spectrum of $H_0$ has an isolated component $\sigma_0$. Since $H_0$ is time reversal symmetric, the projection $P_0$ onto the isolated component $\sigma_0$ admits a Wannier basis $\left\{\w_{\gamma,a}:=\w_{0,{a}}(\cdot -\gamma )\right\}_{\gamma \in \Gamma, 1\leq a \leq m}$, for some $m>0$, where each Wannier function $\w_{\gamma,{a}}$ is exponentially localized around $\gamma$ in the sense of Definition~\ref{GWFLoc}. By standard perturbation theory \cite[Remark VII.2.3, p. 379]{Kato} one can prove that there exists a $\lambda_0>0$ small enough such that, for $|\lambda|<\lambda_0$, the spectrum of $H_\lambda$ has an isolated component $\sigma_{\lambda}$ that varies continuously with $\lambda$ in the Hausdorff distance. Denote by $P_\lambda$ the spectral projection onto $\sigma_\la$. \\ In \cite{NenciuNenciu1993} it is shown how to transport an exponentially localized GWB from the range of $P_0$ to the range of $P_\la$. Let us review here the main steps of the proof rewritten in our setting. First of all, since $H_\lambda$ is an analytic family of type A, the family of projections is analytic in $\lambda$ \cite[Theorem VII.3.1.7]{Kato}, in particular $\|P_\lambda -P_{\lambda'}\| \leq C |\lambda-\lambda'|$ for some positive constant $C$. It follows that, for $|\lambda-\lambda'|$ small enough, there exists a unitary operator which intertwines the two projections and, as a consequence, it unitarily maps an orthonormal basis for the range of $P_\lambda$ into an orthonormal basis for the range of $P_{\lambda'}$. Furthermore, in view of the hypothesis on the Hamiltonian and the gap condition, one can prove that the projections, $P_\lambda$ and $P_{\lambda'}$, are both exponentially localized projection in the sense of Definition~\ref{dfn:ELP}. In addition, using again the Riesz formula, together with the Combes--Thomas rotation and the relative smallness of $W$,
one gets that, for some $\delta>0$ small enough, $\sup_{{\bf{a}} \in \R^d}\|\expo^{-\delta\langle\cdot-{\bf{a}\rangle}} \left(P_{\lambda}-P_{\lambda'}\right) \expo^{\delta\langle\cdot-{\bf{a}\rangle}}\|\leq C |\lambda-\lambda'|$ which implies that the intertwining unitary, explicitly given by the Kato--Nagy formula, preserves the exponential localization, see \cite{NenciuNenciu1993}.
Therefore, starting from an exponentially localized GWB for $P_0$ and iterating the unitary transport a finite number of times, one obtains a GWB for every
$P_\lambda$, with $|\lambda|<\lambda_0$, which is exponentially localized around the same discrete lattice $\Gamma$ as the original one.

Notice that this argument relies on the fact that $P_\lambda$ is a spectral projection onto an isolated component of the spectrum, namely we are assuming that the Fermi energy lies always in a spectral gap. In the setting of random Schr\"odinger operator it is usually considered also the case of the Fermi energy lying in a region of mobility gap. In such a situation, even though the integral kernel of the Fermi projection is not exponentially localized, it might still be possible to construct special orthonormal basis that are localized in space, see for example \cite{GerminetKleinSchenker2007} and references therein. 
\end{example}

\begin{example}[Deformed $d$-dimensional insulators: the Gubanov model] $\,$\newline
\label{example:DeformedTInsulator} 
	A class of examples of non-periodic gapped quantum systems that admit an exponentially localized GWB is provided 
	by a Schr\"odinger operator modeling the deformation of a periodic $d$-dimensional insulator. 
	We recall the definition of this model, following \cite{DeMonvelNenciuNenciu1995}, 
	where it is called \emph{quasi crystalline model of Gubanov}. 
	
	\noindent The Hamiltonian $H_0$ describing the undeformed system, acting in $L^2(\R^d)$ for $d \leq 3$, 
	is given by
	$$
	H_0=-\Delta + V   \qquad  \text{ with } V \in L^2_{\rm uloc}(\R^d). 
	$$
	We assume that the spectrum of $H_0$ has an isolated component $\sigma_0$, and denote by $P_0$
	the corresponding spectral projection. In order to describe the deformation of the system, we consider 
	$g \in C^2(\R^d,\R^d)$ and $\Omega \subset \R^d$ and set
	$$
	\begin{aligned}
	& b(g;\x) = \max_{1\leq i,j,k\leq d} \left[|\partial_{i}\partial_{j}g_k(\x)|\, , |\partial_{i}g_k(\x)|  \right] \\
	& b(g;\Omega)=\sup_{\x \in \Omega} b(g;\x) .
	\end{aligned} 
	$$
	Let $V_g$ be the potential given by $V_g(\x)=V(\x+g(\x))$, for all $\x \in \R^d$. Then the Hamiltonian of the deformed system is  $H_g=-\Delta + V_g$. \\
	Assume that $b(g,\R^d)=\xi < + \infty$ and that $P_0$ admits a GWB which is exponentially localized around some $r$-uniformly discrete set 
	$\Lattice$. A consequence of \cite[Proposition~4]{DeMonvelNenciuNenciu1995} is that, for $\xi$ small enough, the spectrum of $H_g$ has an isolated component $\sigma_{g}$ that varies continuously with $\xi$ in the Hausdorff distance. The corresponding spectral projection  $P_g$ of $H_g$ is not norm continuous  with respect to $\xi$. However if, for $\xi$ small enough, one defines the unitary operator \cite{DeMonvelNenciuNenciu1995} $(Y\psi)(\x)= \left(\det\left(|\partial_{j}(\x+g(\x))_i|\right)\right)^{1/2} \psi(\x+g(\x))$ for every $\psi \in L^2(\R^2)$, then $YP_gY^*$ depends continuously on $\xi$, 
	hence  $\norm{YP_gY^* - P_0} < 1$  for $\xi$ sufficiently small. Therefore, by repeating the steps described in Example \ref{example:Impurties} one can unitarily transport the GWB from the range of $P_0$ to the range of $YP_gY^*$, without spoiling the localization properties. Let $\{\w_{\gamma,a}\}_{\gamma \in \Lattice, 1 \leq a \leq m(\gamma)}$ be the GWB of $YP_gY^*$. Exploiting the explicit expression of the operator $Y^*$ and the fact that the derivatives of $g$ are uniformly bounded by $\xi$, one gets that $\{\varphi_{\widetilde{\gamma},a}:=Y^* \w_{\gamma,a}\}_{\widetilde{\gamma} \in \widetilde{\Lattice}, 1 \leq a \leq m(\widetilde{\gamma})}$ is an exponentially localized GWB for the range of $P_g$, where $\widetilde{\Lattice}:=\{\x \in \R^d \; |\; \x= \gamma+g(\gamma), \gamma \in \Lattice\}$ is an $r'$-uniformly discrete set and $m(\widetilde{\gamma})=m(\gamma)$. 
    \end{example}

\begin{example}[Magnetic perturbations of Chern-trivial 2D insulators]
\label{example:MGWB}$ $\newline 
Generalized Wannier functions naturally appear when a $2$-dimensional system is subjected to a constant magnetic field whose flux through the periodicity cell does not necessarily satisfies a commensurability condition with respect to the quantum of magnetic flux. While it is not possible, in such a case, to rely on any type of magnetic Bloch-Floquet transform, it is still possible to exploit the covariance with respect to magnetic translations. This problem has been tackled in \cite{CoHeNe2015} and \cite{CorneanMonacoMoscolari2019}. 
Consider a $\Z^2$-periodic insulator modelled by a Bloch-Landau Hamiltonian, acting in $L^2(\R^2)$, that is  
\begin{equation}
\label{Bloch-Landau}
H_{\epsilon}=\left( -\iu\nabla - \mathbf{A}_{\Z^2} -b_0\mathbf{A}-\epsilon\mathbf{A} \right)^2 + V_{\Z^2}
\end{equation}
where $V_{\Z^2}$ is a smooth $\Z^2$-periodic scalar potential, $\mathbf{A}_{\Z^2}$ is a smooth $\Z^2$-periodic vector potential, $\epsilon \in \R$, $b_0 \in 2\pi \mathbb{Q}$ and $\mathbf{A}(\x)=\frac{1}{2}\left(-x_2,x_1\right)$ is the magnetic potential of a constant magnetic field in the symmetric gauge. Assume that the spectrum of $H_0$ has 
an isolated component $\sigma_0$ and let $P_0$ be the corresponding spectral projection. 
Then, for $|\epsilon|$ small enough, the spectrum of $H_\epsilon$ has an isolated component $\sigma_{\epsilon}$ which varies continuously with $\epsilon$ in the Hausdorff distance. If $P_0$ has a vanishing Chern number, then $P_0$ admits a (magnetic) Wannier basis $\left\{\w_{\gamma,{a}}\right\}_{\gamma \in \Z^2, 1\leq a \leq m}$, for some $m>0$, where each Wannier function $\w_{\gamma,{a}}$ is exponentially localized around $\gamma$.  
A consequence of the results in \cite{CoHeNe2015} is that, if $|\epsilon|$ is small enough, then $P_\epsilon$ admits a GWB exponentially localized around the same lattice $\Z^2$. Notice that the \virg{magnetic} GWB for the perturbed projection $P_\epsilon$ present an \emph{almost-ladder structure}, in the sense that it can be written just using a finite set of vectors together with magnetic translation and a position dependent phase, as explained in \cite{CorneanMonacoMoscolari2019}, where also a constructive algorithm for the GWB is provided.
\end{example}

%\goodbreak
\newpage

\subsection{A topological marker in position space} $ $\newline
\label{subsec:ChernMarker}

The TKNN approach to the quantum Hall effect  \cite{TKNN82} established that the Hall conductivity, 
as given by the Kubo formula for linear response, is always an integer in units of  $\frac{e^2}{h} = \frac{1}{2\pi}$.  
Shortly after, the topological origin of such an integer has been recognized \cite{AvronSeilerSimon1983,Simon1983}, thus establishing a \emph{Transport-Topology Correspondence}: 
for gapped periodic magnetic systems, the Hall conductivity equals (up to a universal constant) 
the Chern number $c_1(P)$ of the (magnetic) Bloch bundle corresponding to the Fermi projection. 
Later, Haldane noticed that such a correspondence and the existence of a non-trivial topology only
rely on the breaking of TRS, thus paving the way to the flourishing new field of topological insulators 
\cite{Haldane88, HasanKane}. 

The former Transport-Topology Correspondence for periodic systems  
has been later extended to non-periodic models, either by methods from non-commutative geometry \cite{Bellissard1986,NakamuraBellissard,BellissardElstSchulz-Baldes1994}, or by using the index of a pair of projections \cite{Bellissard1986,AvronSeilerSimon1994}.  As mentioned in the Introduction, in the non-periodic setting the Chern number 
is replaced by the \crucial{Chern character} $C(P)$, given by formula \eqref{Chern_marker_intro} \cite{Bellissard1986,BellissardElstSchulz-Baldes1994}.  Its main advantage is that it provides a topological marker defined \emph{in position space}, with no reference to a quasi-momentum space which is available only in a periodic setting.  

We emphasize that the Chern character can be defined without any ergodicity or covariance assumption on $P$, as in the following

\begin{dfn}[Chern character]
\label{def:Chernmarker}
Let $P$ be an orthogonal projection in $L^2(\R^2)$. Let $X_1$ and $X_2$ be the multiplication operators by the respective component of the position operator, \ie $(X_i \w)(\x)=x_i \w(\x)$, for $\w \in \mathcal{D}(X_i)$. 
Given $L>1$ and $\Lambda_L:= [-L,L]^2$, we denote by $\chi_{\Lambda_L}$ the characteristic function of the set 
$\Lambda_L$. Assume that $\left[X_1, P \right], \left[X_2,P \right]$  uniquely extend to bounded operators and $\chi_{\Lambda_L} P \left[\left[X_1, P \right],\left[X_2,P \right]\right]P\chi_{\Lambda_L}$ is a trace class operator for every $L>1$. \newline
Under this assumptions, the \emph{Chern character} of $P$ is defined by the following trace per unit volume
	\begin{equation}
	\begin{aligned}
	\label{ChernMarker}
	C(P)&:= \lim\limits_{L \to \infty} \frac{2\pi}{4L^2} \Tr \left(\iu\chi_{\Lambda_L} 
	P \Big[ \left[X_1, P \right],\left[X_2,P \right] \Big] P \chi_{\Lambda_L}\right)
	\end{aligned}
	\end{equation}
whenever the limit exists.
\end{dfn}

Formula \eqref{ChernMarker} coincides up to a universal constant with the Hall conductivity in gapped systems, in the linear response regime, provided Kubo formula holds true. For the validity of the latter, 
see  the recent papers \cite{MonacoTeufel, Teufel, Marcelli2022, MarcelliMonaco2021} and references therein; 
concerning the Kubo formula for \emph{spin} conductivity, see \cite{MarcelliPanatiTauber, MarcelliPanatiTeufel} and references therein. 

The limit in \eqref{ChernMarker} equals $2\pi$ times the trace per unit volume of the operator 
\begin{equation} \label{frakC_P}
\mathfrak{C}_{P}:=\iu P \left[\left[X_1, P \right],\left[X_2,P \right]\right]P, 
\end{equation}
hence it agrees with formula \eqref{Chern_marker_intro} in the Introduction. It is interesting to notice that it is possible to rewrite \eqref{frakC_P} in terms of commutators of the so called \emph{reduced position operators}. { Let $P$ be an exponentially localized projection} and $\widetilde{X}_j\;:=\; P X_j P$ be the reduced position operator in direction $j \in \{1,2\}$. 
A direct computation\footnote{Notice that all the terms involved in the direct computation are trace class operators, as it is proved in Proposition \ref{TraceClassProof}.}, exploiting only $P^2=P$ and $\left[X_1,X_2\right]\;=\;0$, yields
\begin{equation}
\label{eq:fromCherntoPXPYP}
\chi_{\Lambda_L} \mathfrak{C}_P=\chi_{\Lambda_L}\iu P X_1 P X_2 P - \chi_{\Lambda_L} \iu P X_2 P X_1 P \;=\; \chi_{\Lambda_L} \iu \big[\widetilde{X}_1,\widetilde{X}_2\big],
\end{equation} 
so that 
\begin{equation} \label{Chern_commutator}
C(P) = 2\pi \, \Tuv\left(\iu  \big[\widetilde{X}_1,\widetilde{X}_2\big]  \right), 
\end{equation}
where $\Tuv(\cdot)$ denotes the trace per unit volume, which is \emph{conditionally cyclic}%%%FOOTNOTE
\footnote{Here, we use the adverb \virg{conditionally} to stress the fact that $\Tuv(\cdot)$ is cyclic under some additional conditions, \eg the operators involved in its arguments are periodic. Indeed, even if the Chern character can be expressed as the trace per unit volume of the commutator between the reduced position operators $\widetilde{X}_1$ and $\widetilde{X}_2$, it may very well be non-vanishing (the reduced position operators are not periodic).}%%%%
.
Notice that the operator $\big[\widetilde{X}_1,\widetilde{X}_2\big]$ is densely defined since the projection $P$ maps 
pointwise exponentially decaying functions, with exponential decay less than $\beta$ in \eqref{DiagLocIntKErnel}, into pointwise exponentially decaying functions.

Definition~\ref{def:Chernmarker} contains two relevant constraints: the trace class condition of 
$\chi_{\Lambda_L}\mathfrak{C}_{P}\chi_{\Lambda_L}$ and the existence of the limit of \eqref{ChernMarker}, which are not trivial when $P$ is a generic orthogonal projection. Whenever the orthogonal projection is exponentially localized and time reversal symmetric, the Chern character vanishes, independently of any periodicity hypothesis, \cf \cite[Theorem 3.9]{AvronSeilerSimon1994}, as detailed in the following

\begin{prop}[Chern character vanishes under TRS]
\label{Prop:TRSChern}
Let $P$ be an exponentially localized projection acting on $L^2(\R^2)$. Let time-reversal symmetry be encoded by\footnote{In our setting, a canonical time-reversal operator is simply given by the complex conjugation operator.} an antiunitary operator $\Theta$ such that $\Theta^2=\pm\Id$ and $[\Theta, X_i]=0$ for any $i\in \{1,2\}$. If $P$ is time-reversal symmetric in the sense that $\Theta P \Theta^{-1}= P$, then $C(P)=0$.
\end{prop}
\begin{proof}[Proof]
First, notice that by Proposition \ref{TraceClassProof} below, the sequence of traces defining the Chern character is well defined. Then, by using that $\Tr(A)=\overline{\Tr(\Theta \,A\, \Theta^{-1})}$ for every trace class operator $A$ and exploiting the time reversal symmetry of $P$, we have the chain of equalities
\begin{align*}
&\iu \Tr \left(\chi_{\Lambda_L} P \left[ \left[X_1, P \right],\left[X_2,P \right] \right] P \chi_{\Lambda_L}\right)=\iu\, \overline{\Tr \left(\Theta\chi_{\Lambda_L} P \left[ \left[X_1, P \right],\left[X_2,P \right] \right] P \chi_{\Lambda_L}\Theta^{-1}\right)}\\
&=\iu\, \Tr \left({\left(\chi_{\Lambda_L} P \left[ \left[X_1, P \right],\left[X_2,P \right] \right] P \chi_{\Lambda_L}\right)}^*\right)\\
&=-\iu\, \Tr \left(\chi_{\Lambda_L} P \left[ \left[X_1, P \right],\left[X_2,P \right] \right] P \chi_{\Lambda_L}\right)
\end{align*}
which implies that the Chern character is zero.
\end{proof}

Notice that the exponential localization hypothesis in Proposition \ref{Prop:TRSChern} is only used to justify the formal computations with the double commutator formula; however, a sufficiently fast polynomial decay of the integral kernel 
would be sufficient to prove the statement.

Even though  the limit \eqref{ChernMarker} defining the Chern character might not exist for generic orthogonal projections, this is never the case when $P$ is an exponentially localized projection. Moreover, in such situation, the Chern character is 
an integer. This fact has first been proved in \cite[Appendix, Proposition 3 and Remark 3]{ElgartGrafSchenker} in the discrete setting and using a different terminology%%
\footnote{Notice that in \cite{ElgartGrafSchenker} the authors focus on the Hall conductance, which however for $d=2$, 
as in our setting, equals the Hall conductivity and hence is proportional to the Chern character.}, covering also the more general case of a mobility gap. The proof in \cite{ElgartGrafSchenker} is based on the index of pair of projections defined in \cite{Bellissard1986,NakamuraBellissard,AvronSeilerSimon1994}, and here we adapt their argument to our setting,  \ie exponentially localized projections on the continuum. We notice that the profound reason for the quantization of such index can be traced back to the Fedosov formula for the index of elliptic operators \cite{Fedosov, Hormander}. Before formulating the results, let us briefly remind the definition of the index of pair of projections. 
		
		Consider the unitary multiplication operator $U$, defined by $(U\psi)(\x) =U(\x) \psi(\x)$ with 
		\begin{equation}
		\label{eq:U}
		U(\x):= \left\{ 
		\begin{aligned} &\frac{x_1+\iu x_2}{\sqrt{x_1^2+x_2^2}} & \qquad x_1 + \iu x_2 \notin [0,+\infty) \\
		& 1      &\qquad \, x_1 + \iu x_2  \in [0,+\infty).
		\end{aligned}\right. 
		\end{equation} 
		The unitary $U$ represents a singular gauge transformation associated with the insertion of a magnetic flux tube carrying one unit of quantum flux. 
As emphasized in \cite{AvronSeilerSimon1994},  it is not restrictive to consider only unitary operators 
with unit winding number, like \eqref{eq:U}, see \cite[Theorem 4.2]{AvronSeilerSimon1994}. 
		Then, the index of pair of projections $P$, $UPU^*$, is defined by 
		\begin{equation}
		\label{eq:index}
		\operatorname{Index}(P,UPU^*):=\dim\left(\ker(P-UPU^*-1)\right)-\dim\left(\ker(UPU^*-P-1)\right) \in \Z.
		\end{equation}
		
We show that the Chern character coincides with the index of a pair of projection, and hence is an integer, 
with no ergodicity or covariance assumption on $P$.
		
		\begin{prop}
			\label{prop:CIndex}
			Let $P$ be an orthogonal projection acting on $L^2(\R^2)$ which is an exponentially localized projection in the sense of Definition \ref{dfn:ELP}. Let $U$ be the unitary multiplication operator defined in \eqref{eq:U}. Then,
			\begin{equation}
			C(P)=\operatorname{Index}(P,UPU^*)
			\end{equation}	
			where $C(P)$ is the Chern character defined in \eqref{ChernMarker} and $\operatorname{Index}(P,UPU^*)$ is the index of pair of projections defined in \eqref{eq:index}.
		\end{prop}

The proof is postponed to Appendix \ref{AppendixChern}.

%%%%%%%%%%%%%%%  MAIN RESULT  %%%%%%%%%%%%%%%%%%%%%%%%%%%%%%
\goodbreak

\section{Main result}
We can now formulate our main result.

\begin{thm}[Localization of GWB implies Chern triviality]
	\label{MainTheorem}
	Let $P$ be an orthogonal projection acting on $L^2(\R^2)$, which is an exponentially localized projection in the sense of Definition~\ref{dfn:ELP}.  Let $\Lattice \subset \R^2$ be a $r$-uniformly discrete set and let $s>4$. Suppose that $P$ admits a generalized Wannier basis $\{\w_{\gamma,a}\}_{\gamma \in \Lattice, 1 \leq a \leq m(\gamma)}$ which is $s$-localized around $\Lattice$,
in the sense of Definition~\ref{GWB}. Then the Chern character of $P$ is zero.
\end{thm}

Notice that periodicity plays no role in Theorem~\ref{MainTheorem}. In non-periodic context, one might think that a natural generalization of Bravais lattice is a \emph{Delone set}, which is both $r$-uniformly discrete and $R$-uniformly nowhere sparse$^{\ref{fn:nosparse}}$. However, in the proof of the above Theorem only the property of $\mathfrak D$ to be uniformly discrete is used, via the\virg{Generalized Maclaurin--Cauchy test} (Lemma \ref{prop:SeriesToInt}). 

As we have already mentioned in the Introduction, in the periodic context the relation between the localization of Wannier bases and the topology of the Bloch bundle is a well established paradigm  \cite{MonacoPanatiPisanteTeufel2018}. 
Inspired by the latter result, we conjecture that the remarkable dichotomic character that pertains to periodic systems is actually a more general phenomenon.
\begin{conj}[Localization dichotomy for non-periodic gapped systems]
	\label{ConjectureLocDic}
	\hspace{1pt}
	Let $P_0$ be the spectral projection, acting on $L^2(\R^2)$, corresponding to an isolated component 
	$\sigma_0$ of an operator satisfying the hypotheses of Proposition~\ref{prop:GQS}. Then there exists a number $s_* \geq 1$  such that
the following statements are equivalent:
	\begin{enumerate}[label=(\alph*),ref=(\alph*)]
		\item \label{exp-loc} $P_0$ admits a generalized Wannier basis that is {exponentially localized} 
		around  a $r$-uniformly discrete set, for some $r >0$.
		\item \label{s-loc} $P_0$ admits a generalized Wannier basis that is {$s$-localized, with $s \geq s_*$}, 
		around a $r'$-uniformly discrete set, for some $r' >0$.
		\item \label{c-zero} $P_0$ is {Chern trivial} in the sense that its Chern character $C(P_0)$ 
		exists and is equal to zero.
	\end{enumerate}
\end{conj}

Clearly, \ref{exp-loc} implies \ref{s-loc} by a simple inequality. The implication from \ref{s-loc} (with $s_*>4$) to \ref{c-zero} is the content of Theorem \ref{MainTheorem}, which is proved in the next section. 
The open problem in the conjecture is to show that \ref{c-zero} implies \ref{exp-loc}. Moreover, since for any time reversal symmetric system the Chern character vanishes, see Proposition \ref{Prop:TRSChern}, proving that \ref{c-zero} implies \ref{exp-loc} would solve a long standing conjecture about the existence of generalized Wannier bases for time reversal symmetric systems \cite{NenciuNenciu1993}. \\
Furthermore, in view of the periodic counterpart, the threshold $s_*>4$ appearing in Theorem \ref{MainTheorem}
does not seem optimal. In fact, a generalization of the localization dichotomy for periodic systems 
proved in \cite{MonacoPanatiPisanteTeufel2018} would require $s_*=1$.  

Theorem \ref{MainTheorem}, coupled with the techniques used in \cite{NenciuNenciu1993}, allows to show that the non-existence of a well-localized GWB is stable with respect to suitable perturbations, as detailed in the following Corollary. A similar stability result concerning the vanishing of the Chern character holds also true. 
\begin{crl}[Stability of the GWB delocalization]
\label{Corollary}
Consider the family of Hamiltonian operators, acting in $L^2(\R^2)$,
$$
H_\lambda=\frac{1}{2}(-\iu \nabla - b_0{\bf{A}}-\mathbf{A}_{\Z^2})^2 + V_{\Z^2}+ \lambda W \, ,
\qquad   \lambda\in\R
$$
where $V_{\Z^2}$ is a smooth $\Z^2$-periodic scalar potential, $\mathbf{A}_{\Z^2}$ is a smooth $\Z^2$-periodic vector potential, $b_0 \in 2\pi \mathbb{Q}$ and ${\mathbf{A}}(\x)=\frac{1}{2}\left(-x_2,x_1\right)$. 
Assume that $W$ is infinitesimally relatively bounded with respect to $H_0$ and that the spectrum of $H_0$ has an isolated component $\sigma_0\subset (E_-,E_+)$, $E_\pm \in \R \setminus \sigma(H_0)$. 
For $|\la|$ small enough, say $|\la| < \delta$,  $E_\pm$ belong to the resolvent set of $H_{\la}$. Denote by $P_\la$ the spectral projection onto the spectral island $\sigma_\la:=\sigma(H_\la) \cap (E_-,E_+)$. It holds true that
\begin{enumerate}[label=(\roman*), ref=(\roman*)]
	\item \label{Cor1}if the the Chern number of $P_0$ is different from zero, then for $|\la| < \delta$ the projection $P_\lambda$ does not admit any GWB that is $s$-localized for $s>4$;
	\item \label{Cor2} if $P_0$ admits an $s$-localized GWB around a certain $r$-uniformly discrete set for $s>4$, then $C(P_\la)=0$ for $|\la|<\delta$.
\end{enumerate}	
\end{crl}
\begin{proof} 
Let us prove statement \ref{Cor1}. We first consider $\la =0$. Suppose by contradiction that $P_0$ admits a $s$-localized GWB, with $s>4$. Then, in view of Theorem \ref{MainTheorem}, the Chern character of $P_0$ vanishes. This leads to a contradiction since, as $P_0$ commutes with a unitary representation of $\Z^2$, 
one has $C(P_0) = c_1(P_0) \neq 0$. Consider now $\la \neq 0$.  Since $W$ is infinitesimally small with respect to $H_0$, $H_\lambda$ defines an entire family of type A, hence for $|\lambda|$ small enough $E_\pm$ belongs to the resolvent set of $H_\lambda$ and $\sigma_\lambda$ is a well defined isolated component of the spectrum \cite[Remark VII.2.3, p. 379]{Kato}, such that  $\|P_\lambda - P\| \leq C |\lambda|<1$.
Assume now by contradiction that $P_\lambda$ admits a GWB that is $s$-localized for $s>4$.
Therefore, following the idea in \cite{NenciuNenciu1993}, it is possible to unitarily transport the GWB back to the original system by using the Kato--Nagy unitary operator $U$ that intertwines $P_0$ and $P_\la$.
 Moreover, since $H_\la$ satisfies the assumptions of Proposition \ref{prop:GQS}, we can apply the result of \cite[Lemma C.1]{CorneanMonacoMoscolari2019}, which implies that $U-\Id$ is an integral operator with an exponentially localized kernel, that is a kernel satisfying \eqref{DiagLocIntKErnel}. Hence, via a Schur--H\"olmgren estimate we get that 
$$
\sup_{{\bf{a}} \in \R^2} \| \langle \cdot- {\bf{a}} \rangle^{s} U \langle \cdot- {\bf{a}} \rangle^{-s} \| < +\infty
$$
which implies that $U$ preserves the localization properties of the GWB. Therefore, $P_0$ admits a GWB that is $s$-localized.     
Then, Theorem~\ref{MainTheorem} implies that the  $C(P_0)=0$, which is a contradiction.

The proof of statement \ref{Cor2} follows a similar argument.
\end{proof}

\section{Proofs}
\subsection{Well-posedness of the Chern character}\label{Sec:technical}

The main result of this section is the proof that 
the operator $\chi_{\Lambda_L} \mathfrak{C}_P \chi_{\Lambda_L}$ with $P$ an exponentially localized projection, 
see \eqref{frakC_P}, is trace class, and that its trace is $\Or(L^2)$.

As anticipated in Section \ref{subsec:ChernMarker}, the topological properties of the system are encoded in the commutator $\left[\widetilde{X}_1,\widetilde{X}_2\right]$, and the Chern character is proportional to
\begin{equation}
\label{TraceUnitVolume}
\Tuv \left( \mathfrak{C}_P \right)=\Tuv \left( \iu \left[\widetilde{X}_1,\widetilde{X}_2\right]\right) \;=\; \lim\limits_{L \to + \infty} \frac{1}{4 L^2} \Tr \left(\chi_{\Lambda_L}\, \iu \,\left[\widetilde{X}_1,\widetilde{X}_2\right]\chi_{\Lambda_L}\right).
\end{equation} 

Now we prove that for all $L>1$, each summand of the commutator
$\chi_{\Lambda_L} [\widetilde{X}_1,\widetilde{X}_2]\chi_{\Lambda_L}$, and the terms involved in the direct computation \eqref{eq:fromCherntoPXPYP}, are trace class. 
\begin{prop}
	\label{TraceClassProof}
	Let $P$ be an exponentially localized projection acting on $L^2(\R^2)$, 
	let $X_i$, for $i \in \set{1,2}$, denote the components of the position operator, 
	and let $\chi_{\Lambda_L}$ be the characteristic function of the set $\Lambda_L:=[-L,L]^2$. Then
\begin{enumerate}[label=(\roman*), ref=(\roman*)]
\item \label{item: traceclassofeachsummand} for every $m,n,m',n' \in \N = \set{0,1,\ldots}$
the operator
	$$
	\chi_{\Lambda_L}  P (X_i)^m(X_j)^{m'} P (X_j)^n (X_i)^{n'} 
	$$ 
	is a trace class operator.
\item \label{item: traceclassestdoublecomm}	
	Moreover, there exists $C>0$ such that
	$$
	\left|\Tr(\chi_{\Lambda_L}  P \left[\left[X_1, P\right],\left[X_2, P\right]\right] \chi_{\Lambda_L})\right| 
	\leq C L^2  \qquad \forall L \geq 1.
	$$
\end{enumerate}	
\end{prop}
\begin{proof}
	\ref{item: traceclassofeachsummand}. Consider $\beta$ as in inequality \eqref{DiagLocIntKErnel}. Let $\eta<\beta/3$ and consider the multiplication operators  $\E^{-\eta \| X \|}$ and  $\E^{\eta \| X \|}$. 	One has
	\begin{align*}
	&\chi_{\Lambda_L} P (X_i)^m (X_j)^{m'}P (X_j)^n(X_i)^{n'}\\
	& = \chi_{\Lambda_L} \, \E^{4\eta \|X\| }\, \E^{-4\eta \| X \|} \, P \, \E^{2\eta \|X\| }  \, (X_i)^m (X_j)^{m'}  \,  \E^{-2\eta \| X \|}\, P \, (X_j)^n (X_i)^{n'} \, .
	\end{align*}
	From estimate \eqref{DiagLocIntKErnel} and the triangle inequality it follows that
	\begin{align*}
	&\left|\left(  \E^{-4\eta \| X \|} \, P \, \E^{2\eta \|X\| }  \, (X_i)^m (X_j)^{m'} \right)(\x,\y) \right| \leq  C \E^{-\eta \| \x \|} \E^{-(\beta-3\eta) \|\x-\y\|}   \, ,\\
	&\left| \left(  \E^{-2\eta \| X \|}\, P \, (X_j)^n (X_i)^{n'} \right)(\x,\y) \right| \leq  C' \E^{-(\beta-\eta) \|\x-\y\|}  \E^{-\eta \| \x \|}\, .
	\end{align*}
	Thus, the two operators appearing on the l.h.s above have both integral kernels in $L^2(\R^2 \times \R^2)$, and hence are Hilbert--Schmidt operators. Since the product of two Hilbert--Schmidt operators is in the trace class ideal, and $\chi_{\Lambda_L} \, \E^{4\eta \| X \|}$  is a bounded operator, the first part of the proposition is proved.  In particular, choosing $m=0=m'$ and $n=0=n'$, one concludes that $\chi_{\Lambda_L}P^2=\chi_{\Lambda_L}P$ is trace class.
	 	
	\ref{item: traceclassestdoublecomm}. We have that
    	$$
	\begin{aligned}
	\left|\Tr(\chi_{\Lambda_L}  \mathfrak{C}_P \chi_{\Lambda_L})\right| = \left|\Tr(\chi_{\Lambda_L} \mathfrak{C}_P P \chi_{\Lambda_L})\right| \leq \| \chi_{\Lambda_L} \mathfrak{C}_{P}  P  \chi_{\Lambda_L} \|_1 \leq \|\chi_{\Lambda_L} \mathfrak{C}_{P}  P \|_2 \| P \chi_{\Lambda_L} \|_2 .
	\end{aligned} 
	$$
	By writing explicitly the integral kernel of $\mathfrak{C}_{P}$ and exploiting that the kernel of $P$ is jointly continuous and satisfies \eqref{DiagLocIntKErnel}, we obtain that there exist two positive constants $\alpha$ and $C$ such that
	$$
	|\left(\chi_{\Lambda_L}  \mathfrak{C}_{P}  P \right)(\x,\y)| \leq C \chi_L(\x) \E^{-\alpha\|\x-\y\|} .
	$$
	Hence, we can easily estimate the Hilbert-Schmidt norm of $\chi_{\Lambda_L}  \mathfrak{C}_{P}  P $ and $P  \chi_{\Lambda_L}$ by explicit integration and conclude the proof of the proposition.
\end{proof}
%%%%%%%%%%%%%%%%%%%%%%%%%%%%%%%%%%%%%%%%%

\subsection{The proof in a nutshell}
Let $P$ be an exponentially localized projection that admits a GWB localized around a $r$-uniformly discrete set $\Lattice$, with localization function $G$ (see Definition~\ref{GWB}). In general,  one cannot expect the operators $\widetilde{X}_i$ to be diagonal in the GWB representation, as the expectation value of the position operators in the GWB, namely  
$
\langle \w_{\gamma,{a}}, X_i \w_{\eta,{b}} \rangle  \, , 
$
has non vanishing off-diagonal elements. However, in order to understand the main idea behind our proof, let us suppose for the moment that $P$ admits a GWB made of generalized Wannier functions localized on a $r$-uniformly discrete set $\Lattice$ and with compact and mutually disjoint supports. In this setting, we have
\begin{equation}
\label{eq:Xtdiag}
\langle \w_{\gamma,a}, X_i \w_{\eta,b} \rangle= \langle \w_{\gamma,a}, \widetilde{X}_i \w_{\eta,b} \rangle =: \delta_{\gamma,\eta}\delta_{a,b}f_i(\gamma,a),  \qquad i\in \{1,2\} \,,
\end{equation} in view of the mutually disjoint support property of the GWFs. Then, the operators $\widetilde{X}_1$ and $\widetilde{X}_2$ are diagonal operators in the generalized Wannier basis, which implies that 
$[\widetilde{X}_1,\widetilde{X}_2]=0$. Therefore, we can easily see from \eqref{TraceUnitVolume} that in such oversimplified setting the Chern character vanishes and Theorem \ref{MainTheorem} holds true. 

Hence, we find convenient to introduce the operators $\NP_i$, for $i\in \left\{1,2\right\}$, 
by setting
\begin{equation}
\label{LatticePosition}
\NP_i:= \sum_{\gamma \in \Lattice} \sum_{1\leq a \leq m(\ga)} \gamma_i \ket{\w_{\gamma,a}}  \bra{\w_{\gamma,a}} \, ,
\end{equation}
acting on the maximal domain $\mathcal{D}(\NP_i):= \left\{ \varphi \in L^2(\R^2) \; | \; \sum_{\gamma,a} |\gamma_i|^2 |\langle \w_{\gamma,a}, \varphi \rangle |^2 < +\infty \right\}$. \\
The orthonormality of the GWB implies that
\begin{equation*}
\NP_i P = \NP_i  = P \NP_i \, , \qquad  \NP_i \NP_j = \NP_j \NP_i \, ,
\end{equation*}
and, obviously, one has $\Tuv( \left[\NP_i,\NP_j\right] ) = 0$.
The strategy of the proof of Theorem~\ref{MainTheorem} is to control \virg{how far} the commutator between the reduced position operators is from the commutator between the operators $\NP_i$, when they are localized on the compact region $\Lambda_L$. In particular, we show that 
\begin{equation}
\label{eq:ProofExpl}
\left| \Tr\left(\chi_{{\Lambda_L}} \left(\left[\widetilde{X}_1,\widetilde{X}_2\right]-\left[\NP_i,\NP_j\right] \right)\right) \right| = \mathcal{O}(L) ,
\end{equation}
which implies that 
$$
\Tuv\left(\left[\widetilde{X}_1,\widetilde{X}_2\right]\right)=\Tuv\left(\left[\widetilde{X}_1,\widetilde{X}_2\right]-\left[\NP_i,\NP_j\right]\right)=0.
$$

The proof consists in splitting the contribution of the difference of commutators appearing in \eqref{eq:ProofExpl} in several terms and to show that each of these terms separately goes at most linearly in $L$ as $L \to \infty$. The main ingredients of our analysis are the estimates contained in the Proposition \ref{prop:ImpoEstimates} below which are of two types:
\begin{itemize}
	\item Estimates \eqref{MassIn} and \eqref{MassOut} show that the $L^2$ norm contribution in $\Lambda_L$ (resp. in  $\Lambda^c_L$ ) coming from the GWFs that have a center outside $\Lambda_L$ (resp. inside $\Lambda_L$) is at most of order $L$. Loosely speaking, when we look at the norm of GWFs in a certain region of  the plane, the error we make by considering only the ones with the center in such region grows only like the boundary of the region.
	\item The estimates (\ref{F1}-\ref{F4}) are used to analyze the contribution to the trace in \eqref{eq:ProofExpl} coming from the fact that the reduced position operators $\widetilde{X}_i$ are not diagonal in the generalized Wannier basis, namely \eqref{eq:Xtdiag} does not hold true in general. From the proof of \eqref{F4}, we have that such error is again of order $L$ as $L\to \infty$. 
\end{itemize}

\subsection{Proof of Theorem~\ref{MainTheorem}}
To make the proof as clear as possible, 
we first recollect in the next technical proposition all the important estimates on the GWB that we need for the proof. Notice that in order to estimate the values of certain series over the $r$-uniformly discrete set $\Lattice$ we make use of a generalized Cauchy-Maclaurin test, see Lemma \ref{prop:SeriesToInt},  to obtain an upper bound using Lebesgue integrals over the plane.

\begin{prop}
	\label{prop:ImpoEstimates}
	Let the hypotheses of Theorem~\ref{MainTheorem} be satisfied. Then there exist two positive constants $\mo,\mi$ such that, for every $L\geq 1$, one has
	\begin{align}
	\label{MassOut}
	&	\sum\limits_{\xi \in \Lambda_L , \, c} \left\|\chi_{\Lambda_L^c} \w_{\xi,c}\right\| 
	\leq \mo  \, L \, , &\\
	\label{MassIn}
	&	\sum\limits_{\xi \notin \Lambda_L \, , c} \left\|\vphantom{\chi_{\Lambda_L^c} }\chi_{\Lambda_L} \w_{\xi,c}\right\|   
	\leq \mi \, L \, . &
	\end{align}
Moreover, there exists a function $F: [0,+ \infty) \to [0,+ \infty)$  such that
	\begin{equation}
	\label{FOffDiagonal}
	\sum\limits_{a=1}^{m(\gamma)}\sum\limits_{b=1}^{m(\eta)}\, \left|\inner{\w_{\gamma,a}}{ (X_i-\gamma_i) \w_{\eta,b}} \right| \leq F(\|\gamma-\eta\|)  \qquad \forall\, i\in \{1,2\} \, ,
	\end{equation}
and the following integrability conditions are satisfied 
	\begin{align}
	&\label{F1}
	 \int_{\R^2}d\x \, F(\|\x\|) =: I_1 < \infty \, , \\
	\label{F2}
	&\int_{\R^2}d\x \, F(\|\x\|)^2 =: I_2 < \infty \,, \\
	\label{F3}
	&\int_{\R^2}d\x \, F(\|\x\|)\,\|\x\| =: I_3 < \infty \, .\\
	\label{F4}
	&\lim\limits_{L \to + \infty} \frac{1}{L^2} \int_{\Lambda_L} d \x \, \int_{\R^2 \setminus \Lambda_L} d \y \, F^2(\|\x-\y\|)  \, =\, 0 \, .
	\end{align}
\end{prop}
The proof of Proposition~\ref{prop:ImpoEstimates}, postponed to Appendix~\ref{ProofSLoc}, shows that one can choose 
$F$ in the form $F(\|\x\|):= k_s \langle \x \rangle^{-(s-1)}$. However, we prefer to state the Proposition in the form above, to single out the properties of $F$ that will be used in the proof of Theorem~\ref{MainTheorem}. Indeed, we now prove the statement of Theorem~\ref{MainTheorem} using only the estimates \eqref{MassOut}-\eqref{F4}. For the sake of better readability, in the following the generic series $\sum_{\gamma \in \Lattice\cap \Lambda_L} \sum_{1\leq a \leq m(\ga)} A(\gamma,a)$ and $\sum_{\gamma \in \Lattice} \sum_{1\leq a \leq m(\ga)} A(\gamma,a)$ will be written shortly as $\sum_{\gamma\in \Lambda_L ,a}A(\gamma,a)$ and $\sum_{\gamma,a}A(\gamma,a)$ respectively.

Let us start by noticing that by a simple algebraic manipulation one has
\begin{align*}
&\chi_{\Lambda_L}\left(\left[\widetilde{X}_1,\widetilde{X}_2\right] - \left[\NP_1,\NP_2\right]\right)\\
&= \chi_{\Lambda_L}\left[\left(\widetilde{X}_1-\NP_1\right), \left(\widetilde{X}_2-\NP_2\right)\right] + \chi_{\Lambda_L} \left[\left(\widetilde{X}_1-\NP_1\right), \NP_2 \right] + \chi_{\Lambda_L} \left[ \NP_1, \left(\widetilde{X}_2-\NP_2\right) \right] \\
&=: T_1 + T_2 + T_3 \, .
\end{align*}
In the following we show that each of the traces $\Tr(T_i)$, $i \in \{1,2,3\}$, goes at most linearly in $L$ as $L \to \infty$.

First of all, note that all the terms in the commutators appearing in $\chi_{\Lambda_L} T_j$, $j \in \{1,2,3\}$ are trace class due to Proposition~\ref{TraceClass2Proof}.

Let us start by analysing the trace of $\chi_{\Lambda_L} T_2$. By substituting $\chi_{\Lambda_L} = 1 - \chi_{{\Lambda^c_L}}$ and by exploiting the fact that $P$ is an orthogonal projection and that the GWB is an orthonormal basis of $\Ran(P)$, we obtain
\begin{align*}
&\left|\sum\limits_{\xi,c}\sum\limits_{\gamma,a} \inner{\w_{\xi,c}}{\chi_{\Lambda_L}\w_{\gamma,a}}\inner{\w_{\gamma,a}}{ (X_1-\gamma_1) \w_{\xi,c}}\left(\xi_2- \gamma_2\right) \right|\\
&\leq\sum\limits_{\xi \in \Lambda_L , \, c} \sum\limits_{\gamma,a}   \delta_{\gamma,\xi} \delta_{a,c} | \inner{\w_{\gamma,a}}{(X_1-\gamma_1) \w_{\xi,c}}\left(\xi_2- \gamma_2\right)| \\
&\phantom{ \leq } +  \sum\limits_{\xi \in \Lambda_L , \, c} \sum\limits_{\gamma,a} \left|\inner{\chi_{\Lambda_L^c}\w_{\xi,c}}{\w_{\gamma,a}}\inner{\w_{\gamma,a}}{(X_1-\gamma_1) \w_{\xi,c}}\left(\xi_2- \gamma_2\right)\right| \\
&\phantom{ \leq } +  \sum\limits_{\xi \notin \Lambda_L \, , c} \sum\limits_{\gamma,a} \left|\inner{\chi_{\Lambda_L}\w_{\xi,c}}{\w_{\gamma,a}}\inner{\w_{\gamma,a}}{(X_1-\gamma_1) \w_{\xi,c}}\left(\xi_2- \gamma_2\right)\right|  =:T_{21}+T_{22}+T_{23} \, .
\end{align*}

The first series, namely $T_{21}$, is zero after the summation in $\gamma$. The series $T_{22}$ reads
\begin{equation}
\label{eq:auxEstimate}
\begin{aligned}
&  \sum\limits_{\xi \in \Lambda_L , \, c} \sum\limits_{\gamma,a} \left|\inner{\chi_{\Lambda_L^c}\w_{\xi,c}}{\w_{\gamma,a}}\inner{\w_{\gamma,a}}{(X_1-\gamma_1) \w_{\xi,c}}\left(\xi_2- \gamma_2\right)\right| \\
&\leq \sum\limits_{\xi \in \Lambda_L , \, c}   \| \chi_{\Lambda_L^c}\w_{\xi,c}  \|
\sum\limits_{\gamma \in \Lattice} F(\|\gamma-\xi\|) \|\gamma-\xi\| 
\leq  K_r I_3 \mo \,\, L \, ,
\end{aligned}
\end{equation}
where we have used \eqref{MassOut},  \eqref{FOffDiagonal} and \eqref{FOffDiagonal}, together with Lemma \ref{prop:SeriesToInt} to estimate the series with the integral.  Analogously, by using \eqref{FOffDiagonal}, Lemma \ref{prop:SeriesToInt}, \eqref{MassIn} and \eqref{F3} we get that $\left|T_{23}\right| \leq  K_r I_3 \mi  \,\, L $.
\begin{comment}
\begin{align*}
&\sum\limits_{\xi \notin \Lambda_L \, , c} \sum\limits_{\gamma,a}  |\inner{\chi_{\Lambda_L}\w_{\xi,c}}{\w_{\gamma,a}}\inner{\w_{\gamma,a}}{(X_1-\gamma_1) \w_{\xi,c}}\left(\xi_2- \gamma_2\right)|  \\
&\leq  \sum\limits_{\xi \notin \Lambda_L \, , c}    \|\chi_{\Lambda_L}\w_{\xi,c}\|  
\sum\limits_{\gamma \in \Lattice} F(\|\gamma-\xi\|) \|\gamma-\xi\| \leq K_r I_3 \mi  \,\, L \, ,
\end{align*}
where we have used \eqref{FOffDiagonal}, \eqref{MassIn} and \eqref{F3}. 
\end{comment}

Therefore we have obtained an {upper bound} for the trace of $\chi_{\Lambda_L}T_2$, that is
$
\abs{ \Tr(\chi_{\Lambda_L} T_2)} \leq  C  L$ as $L \to \infty.$ An upper bound on the trace of $\chi_{\Lambda_L} T_3$ is obtained by a similar computation: it goes also at most linearly in $L$ as $L \to \infty$. Hence, both these terms do not contribute to the thermodynamic limit \eqref{TraceUnitVolume}.

It remains to estimate the trace of $\chi_{\Lambda_L}T_1$. By similar computations as the ones above, we can write $\chi_{\Lambda_L}T_1=:R_1+R_2+R_3$, where $R_i$, $i\in \left\{1,2,3\right\}$ are (a posteriori absolutely convergent) series such that $R_3$ contains the localization function $\chi_{\Lambda_L^c}$, $R_2$ contains the localization function $\chi_{\Lambda_L}$ and $R_1$ does not contain neither $\chi_{\Lambda_L}$ nor $\chi_{\Lambda_L^c}$. By using \eqref{MassIn},\eqref{MassOut} and \eqref{F1} in a similar way as in \eqref{eq:auxEstimate}, one can show that $R_2$ and $R_3$ are absolutely convergent series that go at most linearly in $L$, as $L \to \infty$. This means that the only contribution to the limit $L\to \infty$ can come from $R_1$. However, this is not the case as we now show.
Consider the series $R_1$:
\begin{equation}
\begin{aligned}
&R_1:=\sum\limits_{\gamma,a} \sum\limits_{\eta,b} \sum\limits_{\xi \in \Lambda_L , \, c} \delta_{\gamma,\xi} \delta_{a,c} \inner{\w_{\gamma,a}}{ (X_1-\gamma_1) \w_{\eta,b}}\inner{\w_{\eta,b}}{ (X_2-\eta_2) \w_{\xi,c}} \nonumber \\
&\quad - \sum\limits_{\gamma,a} \sum\limits_{\eta,b} \sum\limits_{\xi \in \Lambda_L , \, c}\delta_{\gamma,\xi} \delta_{a,c} \inner{\w_{\gamma,a}}{ (X_2-\gamma_2) \w_{\eta,b}}\inner{\w_{\eta,b}}{ (X_1-\eta_1) \w_{\xi,c}} \, .
\end{aligned}
\end{equation} 
By using Lemma \ref{prop:SeriesToInt} and \eqref{F2}, one easily gets that $R_1$ is absolutely convergent together with a non-optimal estimate (quadratic in $L$) for its sum. Let us now show that also $R_1$ goes at most linearly in $L$.
By using the shorthand notation $
	D(\eta,\xi,b,c):=\inner{\w_{\xi,c}}{ (X_1-\xi_1) \w_{\eta,b}}\inner{\w_{\eta,b}}{ (X_2-\eta_2) \w_{\xi,c}}$, $R_1$ can be written as 
\begin{align*}
R_1&=\sum\limits_{\eta,b} \sum\limits_{\xi \in \Lambda_L , \, c} D(\eta,\xi,b,c) - \sum\limits_{\eta \in \Lambda_L,b} \sum\limits_{\xi,c} D(\eta,\xi,b,c)  \\
\label{FinalSeries}
&=
\sum\limits_{\eta \in \Lattice \setminus \Lambda_L, \, b}\,\, \sum\limits_{\xi \in \Lambda_L , \, c} D(\eta,\xi,b,c) - \sum\limits_{\eta \in \Lambda_L,b} \,\,\sum\limits_{\xi \in \Lattice \setminus \Lambda_L, \, c} D(\eta,\xi,b,c) \, .
\end{align*}

Notice that $|D(\eta,\xi,b,c)| \leq F^2(\|\eta-\xi\|)\,$.
For our purposes, we just need to study the asymptotics of  the absolute value  of the series of $R_1$, namely
\begin{align*}
%&\left|\sum\limits_{\eta \in \Lattice \setminus \Lambda_L, \; b} \,\,\sum\limits_{\xi \in \Lambda_L , \, c} D(\eta,\xi,b,c) - \sum\limits_{\eta \in \Lambda_L,b}\,\, \sum\limits_{\xi \in \Lattice \setminus \Lambda_L, \; c} D(\eta,\xi,b,c) \right|  \\
|R_1|&\leq \sum\limits_{\eta \in \Lattice \setminus \Lambda_L, \; b}\,\, \sum\limits_{\xi \in \Lambda_L , \, c}\left| D(\eta,\xi,b,c)\right| + \sum\limits_{\eta \in \Lambda_L,b}\,\, \sum\limits_{\xi \in \Lattice \setminus \Lambda_L, \; c} \left|D(\eta,\xi,b,c)\right| \\
 &\leq 2 \sum\limits_{\eta \in \Lattice \setminus \Lambda_L}\,\, \sum\limits_{\xi \in \Lambda_L }F^2(\|\eta-\xi\|). 
\end{align*}

%It is clear from the definition \eqref{FEstimate} of $F$ that $ \mathcal{F}(\eta,\xi):=F^2(\|\eta-\xi\|)$ 
From the proof of Proposition \ref{prop:ImpoEstimates} it is clear that $F$
is not in $\ell^1(\Lattice \times \Lattice)$. As we cannot invoke Lebesgue dominated convergence theorem in order to perform the limit $L \to \infty$, we explicitly estimate the series with Lemma \ref{prop:SeriesToInt}, obtaining that
\begin{equation}
\lim_{L \to \infty} \frac{1}{4 L^2}\sum\limits_{\eta \in \Lattice \setminus \Lambda_L}\,\, \sum\limits_{\xi \in \Lambda_L }  F^2(\|\eta-\xi\|) \leq K^2_r \lim\limits_{L \to + \infty} \frac{1}{ 4 L^2} \int_{\Lambda_L} d \x \, \int_{\R^2 \setminus \Lambda_L} d \y \, F^2(\|\x-\y\|)  
\end{equation}
where the limit vanishes in view of \eqref{F4}. Taking into account the last estimate, and all the previous ones, we get \eqref{eq:ProofExpl}, and the proof is concluded.   \qed

%%%%%%%%%%%%%%%%%%%%%%%%%%%%%%%%%%%%%%%%%
\appendix
\section{Proof of Proposition \ref{prop:CIndex}}
\label{AppendixChern}

Before starting with the proof of Proposition \ref{prop:CIndex}, let us recall some known facts about the index 
of a pair of projections. First, let ${\bf p} \in \R^2$ and $U_{\bf p}$ be the multiplication operator associated with the function $U(\cdot - {\bf p})$, clearly $U_0 \equiv U$. As it is proved in \cite{AvronSeilerSimon1994}, we have that, for every ${\bf p} \in \R$
\begin{equation}
\label{eq:TIindex}
\begin{aligned}
&\operatorname{Index}(P,U_{\bf p}PU^*_{\bf p})=\operatorname{Index}(P,UPU^*)=\Tr((P-UPU^*)^3)=\\ &\int_{\R^{6}} d \x d \y d \z P(\x, \y) P(\y, \z) P(\z, \x) \left(1-\frac{U(\x)}{U(\y)}\right)\left(1-\frac{U(\y)}{U(\z)}\right)\left(1-\frac{U(\z)}{U(\x)}\right) .
\end{aligned}
\end{equation}
In particular, the first equality shows that the index \eqref{eq:index} is translation invariant.

A second crucial ingredient of the proof of Proposition \ref{prop:CIndex} is the Connes area formula \cite[Lemma 9.2]{Connes}  regarding the area of the oriented triangle spanned by three points in the plane. For every ${\bf p},\x,\y \in \R^2$, let $\sin({\angle}(\x,{\bf p},\y))$ be the sinus of the angle of view ${\angle}(\x,{\bf p},\y) \in (-\pi,\pi)$ from ${\bf p}$ of $\x$ relative to $\y$ (which is the angle formed by the segment $\y-{\bf p}$ and $\x-{\bf p}$). One can easily check that 
$$
\begin{aligned}
&\left(1-\frac{U(\x-{\bf p})}{U(\y-{\bf p})}\right)\left(1-\frac{U(\y-{\bf p})}{U(\z-{\bf p})}\right)\left(1-\frac{U(\z-{\bf p})}{U(\x-{\bf p})}\right) = \\ 
&-2 \iu\left(\sin{\angle}(\x,{\bf p},\y) +\sin {\angle}(\y,{\bf p},\z) + \sin{\angle}(\z,{\bf p},\x) \right)
=: S({\bf p},\x,\y,\z) .
\end{aligned}
$$
Then, the Connes area formula gives
$$
\begin{aligned}
&\int_{\R^2} d {\bf p} \,\,  S({\bf p},\x,\y,\z) 
%=-2\iu \int_{\R^2} d {\bf p}   \left(\sin({\angle}(\x,{\bf p},\y)) +\sin({\angle}(\y,{\bf p},\z)) + \sin({\angle}(\z,{\bf p},\x)) \right) \\
&= 2 \pi \iu  (\x-\y)\wedge(\y-\z)
\end{aligned}
$$
where $\frac{1}{2} (\x-\y)\wedge(\y-\z)$ is the oriented area of the triangle formed by $\x,\y,\z$.

The proof of Proposition \ref{prop:CIndex} exploits the translation invariance of the index, expressed by \eqref{eq:TIindex}, jointly with the Connes area formula, following essentially the same strategy used in \cite{ElgartGrafSchenker} for discrete models. 

\begin{proof}[Proof of Proposition \ref{prop:CIndex}]
To shorten the notation, let $N:=\operatorname{Index}(P,UPU^*)$ and $N_{\bf p}:=\operatorname{Index}(P_{\bf p},U_{\bf p}PU_{\bf p}^*)$. Moreover, we will denote by $K$ the irrelevant strictly positive constants appearing in the proof. From \eqref{eq:TIindex} we have $N_{\bf p}=N$ for every ${\bf p} \in \R^2$, hence
$$
\begin{aligned}
N &= |\Lambda_L|^{-1} \int_{ \Lambda_L} d {\bf p} \, N_{\bf p}  \\
&=|\Lambda_L|^{-1}\int_{\Lambda_L} d {\bf p} \int_{\R^{2}} d \x \int_{\R^2} d \y \int_{\R^2} d \z \, P(\x, \y) P(\y, \z) P(\z, \x) S({\bf p},\x,\y,\z) \\
&=: |\Lambda_L|^{-1}\int_{\Lambda_L} d {\bf p} \int_{\R^{2}} d \x \, f({\bf p},\x).
\end{aligned}
$$
By suitable estimates on the decay of $f$, we have that  $f$ is integrable both in $\R^2 \times \Lambda_L$ and $\Lambda_L \times \R^2$, see \eqref{fxp}. Therefore, the proof is reduced to show that it is possible to exchange the role of the integral in ${\bf p}$ with the integral in $\x$ 
up to an error of order $L^{-1/2}$, which eventually goes to zero by taking the limit $L \to \infty$. More precisely, by adding and subtracting the same term, we get
$$
\begin{aligned}
\int_{\Lambda_L} d{\bf p} \int_{\R^{2}} d \x f({\bf p},\x) &=
\int_{\R^{2}}  d {\bf p} \int_{\Lambda_L} d\x \, f({\bf p},\x) 
- \int_{\R^2\setminus \Lambda_L} d {\bf p} \int_{\Lambda_L} d \x \, f({\bf p},\x) \\
&\phantom{=}+  \int_{\Lambda_L} d{\bf p} \int_{\R^2 \setminus \Lambda_L} d \x f({\bf p},\x)  
=:\int_{\Lambda_L} d\x \, \int_{\R^{2}}  d {\bf p} \, f({\bf p},\x) + R_1 + R_2 .
\end{aligned}
$$
By using the exponential localization of the integral kernel of $P$, we now prove a decay estimate on $f$; after that, exploiting such estimate, we show that the error terms $R_1$ and $R_2$ are of actually of order $L\sqrt{L}$.

First, by using the trivial inequality $|S({\bf p},\x,\y,\z)|\leq 3$ and the exponential localization of the integral kernel of the projection, we get that $f$ is uniformly bounded, that is $|f({\bf p},\x)|\leq K$, for all ${\bf p},\x \in \R^2$.

Assume now that $|\x-{\bf p}|\geq 1$. We first consider the contribution to the integral in case either $\y$ or $\z$ are outside the ball of radius $|{\bf p}-\x|$ centered in $\x$, denoted by $B_{|{\bf p}-\x|}(\x)$. We can bound such contribution by
$$
\begin{aligned}
&3  \int_{\y \notin B_{|{\bf p}-\x|}(\x) } d \y \int_{\R^2} d \z |P(\x, \y) P(\y, \z) P(\z, \x)| \\
&+ 3 \int_{ \R^2 } d \y \int_{\z \notin B_{|{\bf p}-\x|}(\x)} d \z |P(\x, \y) P(\y, \z) P(\z, \x)| \leq \\
& \leq K_1 \int_{|\y-\x| \geq |{\bf p}-\x|} d \y |P(\x,\y)|  
%\leq K_2 \int_{|\y-\x| \geq |{\bf p}-\x|} d \y \, \eu^{-\beta |\x-\y|} 
\leq K \eu^{-\frac{\beta}{2} |{\bf p}-\x|}
\end{aligned}
$$
where we have used that $P$ has an exponentially localized integral kernel. It remains to control the contribution coming from the points where both $\y$ and $\z$ are inside $B_{|{\bf p}-\x|}(\x)$. By a geometric argument, one can prove that the $\angle(\y,{\bf p},\x)$ and $\angle(\z,{\bf p},\x)$ must be in $(-\frac{\pi}{2},\frac{\pi}{2})$. Then, consider the following two estimates: for $\alpha,\beta \in (-\frac{\pi}{2},\frac{\pi}{2})$
$$
\begin{aligned}
&|\sin(\alpha)+\sin(\beta)-\sin(\alpha+\beta)| = |\sin(\alpha)(1-\cos(\beta))+\sin(\beta)(1-\cos(\alpha))|\leq  \\
&|\sin(\alpha)|^2|\sin(\beta)| + |\sin(\beta)|^2|\sin(\alpha)|
\end{aligned}
$$
and, for every $\y$ such that $|\y-\x| < |{\bf p}-\x|$, we obtain $
|\sin(\angle(\y,{\bf p},\x))| \leq |\y-\x||{\bf p}-\x|^{-1}$. Therefore, we get
\begin{equation}
\label{eq:auxS}
|S({\bf p},\x,\y,\z)| \leq  \frac{|\y-\x|^2|\x-\z| + |\y-\x||\x-\z|^2 }{|{\bf p}-\x|^3} .
\end{equation}
Hence, by using \eqref{eq:auxS} together with the exponential localization of $P$, we obtain 
$$
\begin{aligned}
|f({\bf p},\x)| %& \int_{\R^2} d \y \int_{\R^2} d \z |P(\x, \y)| |P(\y, \z)| |P(\z, \x)| |S({\bf p},\x,\y,\z)| \ \\
%& \leq \int_{\R^2} d \y \int_{\R^2} d \z \, \eu^{-\beta |\x-\y|} \eu^{-\beta |\y-\z|} \eu^{-\beta |\z-\x|}  \frac{|\y-\x|^2|\x-\z| + |\y-\x||\x-\z|^2 }{|{\bf p}-\x|^3} \\
 \leq K \frac{1}{|\x-{\bf p}|^3} \int_{\R^2} d \y \int_{\R^2} d \z \, \eu^{-\frac{\beta}{2}|\x-\y|} \eu^{-\frac{\beta}{2}|\z-\x|} \leq K \frac{1}{|\x-{\bf p}|^3} .
\end{aligned}
$$
For $|\x-{\bf p}|\geq 1$, we have $\frac{\langle |\x-{\bf p}|\rangle}{|\x-{\bf p}|} \leq 2$. Therefore, putting together all the previous estimates, we obtain
\begin{equation}
\label{fxp}
|f({\bf p},\x)|\leq K_2 \chi_{|\x-{\bf p}|\leq 1}(\x,{\bf p}) + K_3 \chi_{|\x-{\bf p}|> 1}(\x,{\bf p}) \langle \x-{\bf p} \rangle^{-3} \leq K \langle \x-{\bf p} \rangle^{-3}  \, .
\end{equation}

By simple estimates and integration exploiting \eqref{fxp}, one can show that the error terms are of order $L\sqrt{L}$ as $L \to \infty$, thus we obtain
$$
\begin{aligned}
	&\int_{\Lambda_L} d{\bf p} \int_{\R^{2}} d \x \int_{\R^2} d \y \int_{\R^2} d z P(\x, \y) P(\y, \z) P(\z, \x) S({\bf p},\x,\y,\z) =\\
	&=\int_{\Lambda_L} d\x \int_{\R^{2}} d{\bf p} \int_{\R^2} d \y \int_{\R^2} d \z P(\x, \y) P(\y, \z) P(\z, \x) S({\bf p},\x,\y,\z) + \mathcal{O}(L\sqrt{L}) .
\end{aligned}
$$
Finally, by performing first the integral with respect to ${\bf p}$ and then using Connes area formula the proof is concluded.
\end{proof}

\section{Technical results}

\subsection{Generalized Maclaurin--Cauchy test}
\label{Sec:technical2}

The proof of Theorem~\ref{MainTheorem} is based on the estimates of series evaluated on points of the discrete set $\Lattice$. Because of that, it is useful to have an efficient tool to estimate the value of the series we are interested in. The next lemma concerns a generalized Maclaurin--Cauchy estimate that serves exactly this purpose.
\begin{lemma}[{Generalized Maclaurin--Cauchy test}]
	\label{prop:SeriesToInt}
	Let $\Lattice \subset \R^2$ be a $r$-uniformly discrete set. Consider a continuous function $$D: \R^2 \to \R \, , $$ such that $|D(\x)| \geq  |D(\y)|$ whenever $|\x|\leq |\y|$. Then, there exists a constant $K_r$, depending on $r$ but independent of $L$, such that for every $L>2r$ it holds that
	\begin{equation}
	\label{SeriesToInt}
	\sum_{\gamma \in \Lattice \cap \Lambda_L } |D(\gamma) | \leq K_r \int_{\Lambda_L}  d\x |D(\x)| \, .
	\end{equation}	
\end{lemma}   
\begin{proof}
	The proof is based on the same argument of the well-known Maclaurin--Cauchy integral test. First of all, by the hypothesis on $\Lattice$, it holds that
	$$
	\Lattice \cap B_{r}(\gamma) = \set{\gamma} \; , \qquad  \forall \gamma \in \Lattice.
	$$
	Fix $\rho=2r$. By hypothesis on $\Lattice$, the number of points of $\Lattice$ such that $|\gamma|<\rho$ is finite, so their contribution to the series is finite, say $K_\rho \in \R_+$. Hence one has
	\begin{equation}
	\label{AuxMC1}
	\sum_{\gamma \in \Lattice\cap \Lambda_L} |D(\gamma) |= K_\rho + \sum_{\gamma \in \Lattice\cap \Lambda_L, |\gamma|\geq \rho } |D(\gamma) |\, .
	\end{equation}
	For every point $\gamma \in \Lattice\cap \Lambda_L $ such that $|\gamma| \geq \rho$, 
	one constructs a square $A_r(\gamma)$ of area $\frac{r^2}{2}$ such that one of its vertices is $\gamma$ and for all $\x\in A_r(\gamma)$ it holds that $|\x| \leq |\gamma|$. (For example, $A_r(\gamma)$ might be the open square of diagonal length equal to $r$ constructed along the line passing through the origin and $\gamma$).
	It is also true that 
	$$
	A_r(\gamma) \cap \Lattice =  \set{\gamma} \, , \qquad A_r(\gamma) \subset \Lambda_L \, .
	$$
	Therefore, we obtain that
	\begin{align*}
	\sum_{\gamma \in \Lattice \cap \Lambda_L, |\gamma|\geq \rho } |D(\gamma) |&= \frac{2}{r^2}\sum_{\gamma \in \Lattice\cap \Lambda_L, |\gamma|\geq \rho } |D(\gamma) | \frac{r^2}{2} \\
	&\leq \frac{2}{r^2} \sum_{\gamma \in \Lattice\cap \Lambda_L, |\gamma|\geq \rho } 
	\int_{ A_r(\gamma)} \; d\x \,\, |D(\x)| \leq \frac{2}{r^2} \int_{ \Lambda_L} d \x  |D(\x)| \, .
	\end{align*}
	Then, considering \eqref{AuxMC1} and $\Lambda_\rho \subseteq \Lambda_L$, we get that
	\begin{equation} \label{Final_est}
	\sum_{\gamma \in \Lattice\cap \Lambda_L} |D(\gamma) |\leq \left(  \frac{r^2 K_\rho }{2}\left(\int_{ \Lambda_\rho} d \x 
	\, |D(\x)|\right)^{-1} + 1\right)  \frac{2}{r^2} \int_{ \Lambda_L} d \x  |D(\x)| \, .
	\end{equation}
	This proves the claim, with the constant $K_r$ given by the bracketed expression exhibited in the l.h.s of \eqref{Final_est}. 
\end{proof}

\begin{rmk}\label{rmk:SeriesToInt}
	\label{rmk:generalizedMacLCTest} The result of Lemma \ref{prop:SeriesToInt}, namely inequality \eqref{SeriesToInt}, still holds true if, instead of the radial monotonicity of the function $D: \R^2 \to \R$, we require only a directional monotonicity, that is $|D(\x)| \geq  |D(\y)|$ whenever $|\x_i|\leq |\y_i|$, for $i=1$ or $i=2$. The strategy of the proof is exactly the same, one just need to replace the square $A_r(\gamma)$ with another suitable $r$-dependent set such that $A_r(\gamma) \subset \Lambda_L$, $|A_r(\gamma)|=\frac{r^2}{2}$, and $|\x_i|\leq |\gamma_i|$ for all $\x \in A_r(\gamma)$.
\end{rmk}

\subsection{Properties of the operators $\NP_i$}

In this section we collect the properties of the operators $\NP_i$ that are used in the proof of Theorem \ref{MainTheorem}.

\begin{lemma} Let $\Lattice$ be a $r$-uniformly discrete set, and $\set{\w_{\ga, a}}_{\gamma \in \Lattice, 1 \leq a \leq m(\gamma)}$ a GWB $G$-localized around $\Lattice$. Then, the operators $\NP_i$ defined by \eqref{LatticePosition} are integral operators. Moreover
	\begin{enumerate} [label=(\roman*),ref=(\roman*)]
		\item \label{item:kGammaexp}  if $G(\|\x\|)=\E^{2\alpha \|\x\|}$ with $\alpha<\beta$ for $\beta$ as in \eqref{DiagLocIntKErnel}, then the integral kernel $\NP_i(\x,\y)$ satisfies 
		$$
		\left| \NP_i(\x,\y) \right| \leq C \E^{-\beta' |\x-\y|} + C' |x_i| \expo^{-\beta' | \x-\y |} 
		$$
		for $C,C' >0$ and $\beta'<\alpha$.
		\item \label{item:kGammapol} if $G(\|\x\|)=\langle \x \rangle^{2s}$ with $s>\frac{3}{2}$, then the integral kernel $\NP_i(\x,\y)$ satisfies
		$$
		\left| \NP_i(\x,\y) \right| \leq 
		C \langle \x-\y \rangle^{-(s-\frac{3}{2}-\epsilon)} + C' |x_i| \, \langle \x - \y \rangle^{-(s-1-\epsilon)}
		$$
		for every $0<\epsilon<s-\frac{3}{2}$.
	\end{enumerate}
\end{lemma}

\begin{proof}
	%%%%%%% First case %%%%%%%%%%%%%%%%%%
	Let us show only the proof of \ref{item:kGammapol}, since the  proof of \ref{item:kGammaexp} is similar and simpler. The formal integral kernel of $\NP_i$ is given by
	$$
	\NP_i(\x,\y)= \sum_{\gamma,a} \gamma_i \w_{\gamma,a}(\x) 
	\overline{\w_{\gamma,a}}(\y) \, .
	$$
	For every fixed $\x,\y \in \R^2$ the sum over the indices $\{\gamma,a\}$ is absolutely convergent. Indeed, fix $0<\epsilon<s- \frac{3}{2}$, since $\left\|\x-\gamma \right\| \langle \x - \gamma \rangle^{-1}  \leq 1$, we get that
	\begin{equation}
	\label{NPIntKernel_SLoc}
	\begin{aligned}
	\left| \NP_i(\x,\y) \right| &\leq \sum_{\gamma,a} \langle \x - \gamma \rangle^{-(s-1)} \langle \y - \gamma \rangle^{-s}  + \sum_{\gamma,a} |x_i| \langle \x - \gamma \rangle^{-s} \langle \y - \gamma \rangle^{-s}  \\
	& \leq C_s \langle \x - \y \rangle^{-(s-\frac{3}{2}-\epsilon)}    \sum_{\gamma,a} \langle \y - \gamma \rangle^{-(\frac{3}{2}+\epsilon)} \langle \x - \gamma \rangle^{-(\frac{1}{2}+\epsilon)} \\
	&\phantom{\leq} + C_s |x_i| \, \langle \x - \y \rangle^{-(s-1-\epsilon)} \sum_{\gamma,a}   \langle \y - \gamma \rangle^{-(1+\epsilon)} \langle \x - \gamma \rangle^{-(1+\epsilon)}  \\
	&\leq C \langle \x-\y \rangle^{-(s-\frac{3}{2}-\epsilon)} + C' |x_i| \, \langle \x - \y \rangle^{-(s-1-\epsilon)} \, ,
	\end{aligned}
	\end{equation}
	where again in the first inequality we have used $L^\infty$ estimate on the GWF \eqref{LinftyEstimate}, in the second inequality we have used the property \eqref{eqn:GTriang} of the localization function $G$ (denoting by $C_s$ the constant $C_G$ appearing in \eqref{eqn:GTriang}) and in the last inequality we have used H\"older's inequality. Hence the operators $\NP_i$ admit an integral kernel.
\end{proof}
%%%%%%%%%%%%%%%%%%%%%%%%%%%%%%%%%%%%%%%%%%%
As a consequence of the previous estimates on the integral kernels, the operators $\NP_i$ enjoy some trace class properties. 
\begin{prop} 
	\label{TraceClass2Proof}
	Let $P$ be an exponentially localized projection and $X_i$, $\NP_i$, for $i\in \left\{1,2\right\}$ and $\chi_{\Lambda_L}$ as defined above. 
	Assume that $P$ admits a GWB,  $G$-localized around a $r$-uniformly discrete set $\Lattice$,  
	with localization function $G(\|\x\|) = \jp{\x}^{2s}$ for some $s>\frac{7}{2}$. Then, for every $i,j \in \{1,2\}$ the operators
		$$
		\chi_{\Lambda_L} P X_i P \; \Gamma_j \; \; , \chi_{\Lambda_L} P \; \Gamma_i P \; X_j \; \; , \chi_{\Lambda_L} P \; \Gamma_i P \; \Gamma_j \; ,
		$$ 
		are trace class in $L^2(\R^2)$.
\end{prop}

\begin{proof}
	The strategy of the proof is the same of Proposition~\ref{TraceClassProof}. Let us show the computations explicitly for the operator $\chi_{\Lambda_L}\; P \;\NP_1 \; P \; X_2$, the other cases can be treated similarly.
	
	Since $\Gamma_1$ commutes with $P$, and $P^2=P$, we have that
	\begin{equation*}
	\begin{aligned}
	&\chi_{\Lambda_L}\; P \;\NP_1 \; P \; X_2 = \chi_{\Lambda_L}\; P \;P \;\NP_1 \; X_2 = \chi_{\Lambda_L}\; P \;\NP_1 \; X_2 \;   \\
	& =  \chi_{\Lambda_L}\; \E^{2\alpha\|X\|}\; \E^{-2\alpha\|X\|} \;P\; \E^{\alpha\|X\|}\; \E^{-\alpha\|X\|}\; \NP_1 \; \langle X \rangle\; \langle X \rangle^{-1} \; X_2  \, ,
	\end{aligned}
	\end{equation*}
	where we have chosen the constant $\alpha$ strictly smaller than the exponent $\beta $ appearing in \eqref{DiagLocIntKErnel}. Then, by using the estimate \eqref{DiagLocIntKErnel} and triangular inequality, we have that
	\begin{align*}
	\left| \left( \E^{-2\alpha\|X\|}\; P  \; \E^{\alpha\|X\|} \right) (\x,\y) \right| \leq C \E^{-\alpha \|\x\|} \E^{-(\beta-\alpha) \|\x-\y\|}
	\end{align*}
	for some positive constants $C$. Therefore $\left( \E^{-2\alpha\|X\|}\; P  \; \E^{\alpha\|X\|} \right) $ is a Hilbert--Schmidt operator. Similarly, considering \eqref{NPIntKernel_SLoc} instead of \eqref{DiagLocIntKErnel}, we have
	$$
	\left| \left( \E^{-\alpha\|X\|} \;\NP_1\; \langle X \rangle \right) (\x,\y) \right| \leq  \E^{-\tilde{\alpha} \|\x\|} C \langle \x-\y \rangle^{-(s-\frac{5}{2}-\epsilon)} \,,
	$$
	where $0<\tilde{\alpha}<\alpha$. Since $s>\frac{7}{2}$, we can choose $\epsilon$ small enough so that the integral kernel of $ \left( \E^{-\alpha\|X\|} \;\NP_1\; \langle X \rangle \right)$ is in $L^2(\R^2\times\R^2)$ and hence the operator is Hilbert--Schmidt. Eventually, since $\langle X \rangle^{-1} \; X_2$ and $\chi_{\Lambda_L}\; \E^{2\alpha\|X\|}$ are bounded operators and the product of two Hilbert--Schmidt operators is trace class, the proof is over.
\end{proof}

\begin{rmk}
	Although the proofs of Proposition~\ref{TraceClassProof} and Proposition~\ref{TraceClass2Proof} appear similar, their hypotheses have a crucial difference. Proposition~\ref{TraceClassProof} requires only the localization of the integral kernel of the projection, which follows from the existence of a gap in the spectrum of $H$, while Proposition~\ref{TraceClass2Proof} requires, beyond the gap assumption, the existence of a GWB for the projection with a particular decay of the GWFs. 
\end{rmk}

\subsection{Proof of Proposition~\ref{prop:ImpoEstimates}}
\label{ProofSLoc}
Consider $\Lambda_L = [-L,L]^2 \subset \R^2$ and $\gamma \in \Lattice \cap \Lambda_L$. 
To estimate \virg{how much of $\w_{\gamma,a}$ is outside $\Lambda_L$}  we consider
$\left\| \chi_{\Lambda^c_L} \w_{\gamma,a} \right\|$
where $\chi_{\Lambda_L^c}$ is the  characteristic function of the complementary set of $\Lambda_L$. 
We get
\begin{equation}
\begin{aligned}
\label{MassOutSL}
\norm{\chi_{\Lambda_L^c}\w_{\gamma,a}} &%=
%\left(\int_{\R^2} \chi_{\Lambda_L^c}(\x) \jp{\x-\gamma}^{2s} \jp{\x-\gamma}^{-2s} |\w_{\gamma,a}(\x)|^2  d\x \right)^{\frac{1}{2}}\\
%&\leq 
\leq \left( \sup_{\x \in \R^2} \left( \chi_{\Lambda_L^c}(\x) \jp{\x-\gamma}^{-2s} \right)\int_{\R^2} \jp{\x-\gamma}^{2s} |\w_{\gamma,a}(\x)|^2  d\x \right)^{\frac{1}{2}}  \\
%&\leq M^{\frac{1}{2}} \max\left\{\jp{L-|\gamma_1|}^{-2s},\jp{L-|\gamma_2|}^{-2s}\right\} \\
&\leq M^{\frac{1}{2}}\left(\jp{L-|\gamma_1|}^{-2s} + \jp{L-|\gamma_2|}^{-2s}\right) \, .
\end{aligned}
\end{equation}
Using Remark \ref{rmk:SeriesToInt}, one can see by explicit integration that, if $s>\frac{1}{2}$, then both terms on the right hand side of \eqref{MassOutSL} are integrable in $\R$ and therefore there exists a positive constant $\mo$ such that
$
\sum\limits_{\xi \in \Lambda_L , \, c} \left\|\chi_{\Lambda_L^c} \w_{\xi,c}\right\| \leq  \mo L 
$
.

If, instead, $\gamma$ is outside $\Lambda_L$ we have that
\begin{equation}
\begin{aligned}
\label{MassInSL}
\norm{\vphantom{\chi_{\Lambda_L^c} }\chi_{\Lambda_L} \w_{\gamma,a}} %&=\left(\int_{\R^2} \chi_{\Lambda_L}(\x) |\w_{\gamma,a}(\x)|^2  d\x \right)^{\frac{1}{2}} \\
&\leq \left( \sup_{\x \in \R^2} \left[ \chi_{\Lambda_L}(\x) \jp{\x-\gamma}^{-2s}\right]\int_{\R^2} \jp{\x-\gamma}^{2s} |\w_{\gamma,a}(\x)|^2  d\x \right)^{\frac{1}{2}} \\
& \leq M^{\frac{1}{2}}\left( \jp{|\gamma_1|-L}^{-2s}\chi_{A_1}(\gamma) + \jp{|\gamma_2|-L}^{-2s}\chi_{A_2}(\gamma) \right)  \\
&\phantom{ \leq } + M^{\frac{1}{2}}\jp{\sqrt{||\gamma_1|-L|^2 +||\gamma_2|-L|^2}}^{-2s} \chi_{A_3}(\gamma) \,,
\end{aligned}
\end{equation}
where $A_1:=([-\infty,-L]\cup[L,\infty])\times [-L,L]$, $A_2:= [-L,L]\times([-\infty,-L]\cup[L,\infty])$ and $A_3:=\Lambda_L^c \setminus (A_1 \cup A_2)$, and $\chi_{A_i}$, with $i \in \left\{1,2,3\right\}$, is the characteristic function of the set $A_i$. As in the previous case, using Lemma \ref{prop:SeriesToInt} and Remark \ref{rmk:SeriesToInt}, one can see by explicit integration that, if $s>\frac{1}{2}$, then the first two terms on the right hand side of \eqref{MassInSL} are integrable in $\R$ and, if $s>1$, the last term is integrable in $\R^2$. Therefore there exists a positive constant $\mi$ such that
$
\sum\limits_{\xi \notin \Lambda_L \, , c} \|\chi_{\Lambda_L} \w_{\xi,c}\|  \leq  \mi L 
$
.

It remains to prove the second part of Proposition~\ref{prop:ImpoEstimates}, concerning the off-diagonal terms of $\widetilde{X}_i$. Considering the generic matrix element of $\widetilde{X}_1-\NP_1$, one has
\begin{align}
\nonumber
\sum\limits_{a=1}^{m(\gamma)}\sum\limits_{b=1}^{m(\eta)}|\inner{\w_{\gamma,a}}{ (X_1-\gamma_1)\w_{\eta,b}}| 
\nonumber
&\leq(m_*)^2 \max_{1\leq a \leq m(\gamma)} \max_{ 1\leq b \leq m(\eta)}|\inner{\w_{\gamma,a}}{ (X_1-\gamma_1) \w_{\eta,b}}|  \\
\nonumber 
&\leq  (m_*)^2  \int_{\R^2} d\x \, \left|\w_{\gamma,\tilde{a}}(\x)\right|  |x_1-\gamma_1| \left| \w_{\eta,\tilde{b}}(\x)\right| ,
\end{align}
where $\tilde{a}$ and $\tilde{b}$ are the maximizers of $|\inner{\w_{\gamma,a}}{ (X_1-\gamma_1) \w_{\eta,b}}|$. Therefore we get that
\begin{align*}
&\langle \gamma- \eta \rangle^{s-1}  (m_*)^2  \int_{\R^2} d\x  \,\left| \w_{\gamma,\tilde{a}}(\x) \right| \,  |x_1-\gamma_1| \left| \w_{\eta,\tilde{b}}(\x) \right| \\
&\leq C_s (m_*)^2  \int_{\R^2} d\x \,\langle \gamma - \x \rangle^{s-1} \left| \w_{\gamma,\tilde{a}}(\x) \right|  \langle \x-\gamma \rangle \langle \x - \eta \rangle^{s-1}  \left| \w_{\eta,\tilde{b}}(\x) \right|  \\
%&\leq C_s (m_*)^2  \int_{\R^2} d\x \,\langle \gamma - \x \rangle^{s} \left| \w_{\gamma,\tilde{a}}(\x) \right|   \langle \x - \eta \rangle^{s}  \left| \w_{\eta,\tilde{b}}(\x) \right|  \\
&\leq C_s (m_*)^2  M \, ,
\end{align*}
where in the last inequality we have used Cauchy-Schwarz inequality together with property \eqref{GWFLoc}. This implies that
$$
(m_*)^2  \int_{\R^2} d\x  \,\left| \w_{\gamma,a}(\x) \right| \,  |x_1-\gamma_1| \left| \w_{\eta,{b}}(\x) \right| \leq k_s \langle \gamma- \eta \rangle^{-(s-1)}. 
$$

As clear from above computation,  $k_s$ depends, actually, also on $m_*, K, C_s$. 
The same computation goes through exchanging $X_1-\gamma_1$ with $X_2-\gamma_2$. 
Therefore, the function $F$ defined by
\begin{equation}
\label{FEstimate}
F(\|\x\|):= k_s \langle \x \rangle^{-(s-1)}
\end{equation}
satisfies \eqref{FOffDiagonal}. A direct computation shows that $F$ satisfies also the requirements \eqref{F1} for every $s>3$, \eqref{F2} for every $s>2$ and \eqref{F3} for every $s>4$.

Finally, we show that \eqref{F4} is satisfied for every $s>3$. Indeed, consider the integral
\begin{equation}
\label{Finale1SLoc}
\int_{\Lambda_L} d \x \int_{\R^2\setminus \Lambda_L} d\y \, \frac{1}{(1+ \|\x-\y\|^2)^{(s-1)}} \, .
\end{equation} 
Since the integrand is positive the order of integration does not affect the result. For a fixed $\x \in \Lambda_L$ 
and $\alpha:=s-1$, we have the inequality 
\begin{equation}
\label{eq:aux2}
\frac{1}{(1+ \|\x-\y\|^2)^\alpha} \leq  \frac{1}{(1+ \|\x-\y\|^2)^\frac{\alpha}{2}} \frac{1}{(1+ (\dist(\x,\partial\Lambda_L))^2)^\frac{\alpha}{2}} \, .
\end{equation}
Since $s>3$, the right hand side in \eqref{eq:aux2} is integrable in $\y$, hence we obtain
$$
\begin{aligned}
&\int_{\Lambda_L} d \x \int_{\R^2\setminus \Lambda_L} d\y \,  \frac{1}{(1+ \|\x-\y\|^2)^\alpha}  
\leq C \int_{\Lambda_L} d \x \frac{1}{(1+ (\dist(\x,\partial\Lambda_L))^2)^\frac{\alpha}{2}}
\end{aligned}
$$
for some positive constant $C$. By an elementary estimate and explicit integration one can see that the integral with respect to $\x$ is of order $L$. This concludes the proof. \qed
%%%%%%  BIBLIOGRAPHY  %%%%%%%%%%%%%%%%%%%%%%%%%%%%%

%%%%%%%%  AFFILIATIONS %%%%%%%%%%%%%%%%%%%%%%%%%%%

\vspace{20mm}

{\footnotesize  %%%%%%%%%%%%%%%%%%%%%%%%%%%%%%%%%%%%

	\begin{tabular}{ll}
		
        	(G. Marcelli) 
        	&  \textsc{Mathematics Area, SISSA} \\ 
        	&   Via Bonomea 265, 34136 Trieste, Italy \\
        	&  {E-mail address}: \href{mailto:giovanna.marcelli@sissa.it}{\texttt{giovanna.marcelli@sissa.it}}\\
        	\\
		(M. Moscolari) 
		&  \textsc{Fachbereich Mathematik, Eberhard Karls Universit\"{a}t T\"{u}bingen} \\
		&   Auf der Morgenstelle 10, 72076 T\"{u}bingen, Germany \\
		&  {E-mail address}: \href{mailto:massimo.moscolari@mnf.uni-tuebingen.de}{\texttt{massimo.moscolari@mnf.uni-tuebingen.de}} \\
		\\
		(G. Panati) 
		&  \textsc{Dipartimento di Matematica, \virg{La Sapienza} Universit\`{a} di Roma} \\
		&  Piazzale Aldo Moro 2, 00185 Rome, Italy \\
		&  {E-mail address}: \href{mailto:panati@mat.uniroma1.it}{\texttt{panati@mat.uniroma1.it}} \\
		\\

	\end{tabular}
	
}%%%Endfootnotesize  %%%%%%%%%%%%%%%%%%%%%%%%%%%%%%%%%%%

\begin{thebibliography}{OOOO}
	
	
	
	\bibitem[AW]{AizenmanWarzel}
	\textsc{Aizenman, M., Warzel, S.} : {\it Random operators.} Graduate Studies in Mathematics, American Mathematical Society, 2015
	
	\bibitem[AK]{AucklyKuchment}
	\textsc{Auckly, D., Kuchment, P.}: On Parseval frames of exponentially decaying composite Wannier functions.
	{\it Mathematical Problems in Quantum Physics}, F. Bonetto, D. Borthwick, E. Harrell, and M. Loss (eds.),
	pp. 227–240. Vol. 717 in Contemporary Mathematics Volume. American Mathematical Society (2018)
		
	\bibitem[AHS]{AvronHerbstSimon1981}
	\textsc{Avron, J.~E., Herbst, I.~W., Simon, B.}:
	Schr\"odinger operators with magnetic fields.
    {\it Commun. Math. Phys.} \textbf{79}, 529 (1981)
    
    \bibitem[AS$^2_1$]{AvronSeilerSimon1983}
    \textsc{Avron, J.~E., Seiler, R., Simon, B.}: Homotopy and quantization in condensed matter physics. 
    { \it Phys. Rev. Lett.} {\bf 51}, 51 (1983)
    
    \bibitem[AS$^2_2$]{AvronSeilerSimon1994}
    \textsc{Avron, J.~E., Seiler, R., Simon, B.}: Charge deficiency, charge transport and comparison of dimensions.
    {\it Commun. Math. Phys.} \textbf{159}, 399 (1994)
    
	

	
	
\bibitem[Be]{Bellissard1986}
	\textsc{Bellissard, J.}: Ordinary  quantum  Hall  effect and non-commutative cohomology.  In:  {\it Localization in  disordered  systems. } Weller,  W.,  Zieche,  P.  (eds.),  Leipzig:  Teubner   (1986)
	
			

	\bibitem[BES]{BellissardElstSchulz-Baldes1994}
	\textsc{Bellissard, J., van Elst, A., Schulz-Baldes, H.}: The noncommutative geometry of the quantum {H}all effect.
	{\it J. Math. Phys.} \textbf{35}, 5373 (1994)
	
	\bibitem[BR]{BiancoResta2011}
	\textsc{Bianco, R., Resta, R.}: Mapping topological order in coordinate space. {\it Phys. Rev. B} {\bf 84}, 241106 (2011)
	
	\bibitem[BM]{BourneMesland}
	\textsc{Bourne, C.; Mesland, B.}: Localised module frames and Wannier bases from groupoid Morita equivalences. 
	{\it J. Fourier. Anal. Appl. }{\bf  27}, 69 (2021).  Revised version of:  Gabor frames and Wannier bases from 
	groupoid Morita equivalence, \textsl{arXiv:2009.13806} (2020)
		
	\bibitem[BNN]{DeMonvelNenciuNenciu1995}
	\textsc{Boutet De Monvel-Berthier, A., Nenciu, A., Nenciu, G.}:,
	Perturbed periodic Hamiltonians: essential spectrum and exponential decay of eigenfunctions.
	{\it Lett. Math. Phys.} \textbf{34}, 119 (1995)
	
	\bibitem[BECB]{BradlynEtAL}
	\textsc{Bradlyn, B.; Elcoro, L.; Cano, J.; Vergniory, M. G.; Wang, Z.; Felser, C.; Aroyo,  M. I.; Bernevig, A.}: Topological quantum chemistry. {\it Nature} {\bf 547}, 298--305 (2017)
	
	\bibitem[BHL]{BroderixHundertmarkLeschke2000}
	\textsc{Broderix, K., Hundertmark, D., Leschke, H.}: 
	Continuity properties of Schr\"odinger semigroups with magnetic fields. {\it Rev. Math. Phys.} {\bf 12}, 181--225 (2000)
	
	
	
        \bibitem[BPCM]{BrouderPanati2007}
	\textsc{Brouder Ch.; Panati G.; Calandra M.; Mourougane Ch.; Marzari N.}:  
	Exponential localization of Wannier functions in insulators. 
	{\it Phys. Rev. Lett.} \textbf{98}, 046402 (2007)

	\bibitem[CMCB]{Caio et al 2019}
	\textsc{Caio, M.\,D., M\"oller, G., Cooper, N.\,R., Bhaseen, M.\,J.}: Topological marker currents in Chern insulators. 
	{\it Nature Physics} {\bf 15}, 257--261 (2019)
		
	\bibitem[CLPS]{CaLePaSt2016}
	\textsc{Canc\`{e}s, \'{E}.;  Levitt, A.; Panati, G.; Stoltz, G.}: Robust determination of maximally-localized Wannier functions. {\it Phys. Rev. B} {\bf 95}, 075114 (2017) 
	
	\bibitem[CTVR]{CeresoliThonhauserVanderbiltResta2006}
	\textsc{Ceresoli, D.; Thonhauser, T.; Vanderbilt, D.; Resta, R.}: Orbital magnetization in crystalline solids: Multi-band insulators, Chern insulators, and metals. {\it Phys. Rev. B}\ {\bf 74} 024408 (2006)
	
	\bibitem[Cl$_1$]{Cloizeaux1964} \textsc{des Cloizeaux, J.} : 
	Energy bands and projection operators in a crystal: analytic and asymptotic properties. 
	{\it  Phys. Rev.}\ {\bf  135},  A685--A697 (1964)
	
	\bibitem[Cl$_2$]{Cloizeaux1964a} \textsc{des Cloizeaux, J.} : 
	Analytical properties of $n$-dimensional energy bands and Wannier functions. 
	{\it Phys. Rev.}\ {\bf 135}, A698--A707 (1964)
	
	\bibitem[CT]{CT73}
	\textsc{Combes, J.M.; Thomas, L.}: Asymptotic behavior of eigenfunctions for multiparticle Schr\"odinger operators. 
	{\it Commun. Math. Phys.} {\bf 34}, 251--270 (1973)
	
	\bibitem[Con]{Connes}
	\textsc{Connes, A.}: Noncommutative differential geometry. 
	{\it Inst. Hautes Études Sci. Publ. Math.} {\bf 62}, 257–360 (1985)
	
	\bibitem[CHN]{CoHeNe2015}
	\textsc{Cornean, H.D.; Herbst, I.; Nenciu, G.} : On the construction of composite Wannier functions. 
	{\it Ann. Henri Poincar\'{e}} {\bf 17}, 3361--3398  (2016)
	
	\bibitem[CM]{CorneanMonaco17}
	\textsc{Cornean, H.D., Monaco, D.} : On the construction of Wannier functions in topological insulators: the 3D case. {\it Ann. Henri Poincar\'{e}} {\bf 18}, 3863--3902 (2017)
		
	\bibitem[CMM$_1$]{CorneanMonacoMoscolari2018}
	\textsc{Cornean, H.D., Monaco, D., Moscolari, M.}:
	Beyond Diophantine Wannier diagrams: gap labelling for Bloch-Landau Hamiltonians. 
	{\it J. Eur. Math. Soc.} {\bf 23}, 3679--3705 (2021)

	\bibitem[CMM$_2$]{CorneanMonacoMoscolari2019}
	\textsc{Cornean, H.D., Monaco, D., Moscolari, M.}:
	Parseval frames of exponentially localized magnetic Wannier functions. 
	{\it Commun. Math. Phys.} {\bf 371}, 1179--1230 (2019)
	
		
	\bibitem[CMT]{CorneanMonacoTeufel17}
	\textsc{Cornean, H.D.; Monaco, D.; Teufel, S.}: Wannier functions and $\Z_2$ invariants in time-reversal symmetric topological insulators. {\it Rev. Math. Phys.} {\bf 29}, 1730001 (2017)
	
	\bibitem[CN]{CorneanNenciu09}
	\textsc{Cornean, H.D.; Nenciu, G.}: 
	The Faraday effect revisited: Thermodynamic limit. {\it J.~Funct. Anal.} {\bf 257}, 2024--2066 (2009)
	
	\bibitem[Cos]{Costa2014}
	\textsc{Costa, M.}: {\it Funzioni di Wannier associate ad operatori di Schr\"odinger con un gap nello spettro}. Master Thesis (Supervisor: G. Panati), \virg{La Sapienza} University of Rome, 2014
	

	
	
	\bibitem[CNN]{CorneanNenciuNenciu2008}
	\textsc{Cornean, H.~D., Nenciu, A., Nenciu, G.}: Optimally localized {W}annier functions for quasi one-dimensional
	nonperiodic insulators.
	{ \it J. Phys. A} \textbf{41}, 125202 (2008)
	
	\bibitem[EGS]{ElgartGrafSchenker}
	\textsc{Elgart, A.; Graf, G.; Schenker, J.} 
	Equality of the bulk and edge Hall conductances in a mobility gap. {\it Commun. Math. Phys.} {\bf 259}, 
	185--221 (2005)
	

	\bibitem[Fe]{Fedosov}
	\textsc{Fedosov,  B. V.}: Direct proof of the formula for the index of an elliptic system in Euclidean space. {\it Funktional. Anal. i.  Prilozhen}  {\bf 4}(4), 83--84  (1970);  also in {\it Functional Anal.  Appl.} {\bf 4}, 339-341  (1970)

	\bibitem[FMP]{FiMoPa_2}
	\textsc{Fiorenza, D.; Monaco, D.; Panati, G.} : 
	Construction of real-valued localized composite Wannier functions for insulators. 
	{\it Ann. Henri Poincar\'{e}}  {\bf 17}, 63--97 (2016)
	
	\bibitem[GKS]{GerminetKleinSchenker2007}
	\textsc{Germinet, F.; Klein, A.;  Schenker, J.H.}: Dynamical delocalization in random Landau Hamiltonians. 
	{\it Ann. Math.} {\bf 166} (2007)
	

	\bibitem[Go]{Goedecker} \textsc{Goedecker, S.} :  Linear scaling electronic structure
	methods. {\it Rev. Mod. Phys.} \textbf{71}, 1085--1111 (1999)
		
	
	\bibitem[Gr]{Graf2007}
	\textsc{Graf, G.M.}: Aspects of the integer quantum Hall effect. {\it Proceedings of
		symposia in pure mathematics, spectral theory, and mathematical physics: a
		Festschrift in honor of Barry Simon's 60th birthday}, 429-442,
	Proc. Sympos. Pure Math. {\bf 76}, Part 1, Amer. Math. Soc. (2007)
	
	
	\bibitem[Hal]{Haldane88} 
	\textsc{Haldane, F.D.M.} :  Model for a Quantum Hall
	effect without Landau levels: condensed-matter realization of the  \virg{parity anomaly}. 
	{\it Phys. Rev. Lett.} {\bf 61}, 2017--2020 (1988)
	
	\bibitem[HK]{HasanKane}
	\textsc{Hasan, M.Z.; Kane, C.L.} :  Colloquium: Topological Insulators. 
	{\it Rev. Mod. Phys.} {\bf 82}, 3045--3067 (2010)	
	
	
	\bibitem[HL]{HastingsLoring10}
	\textsc{Hastings, M.B., Loring,  T. A.}: Almost  commuting  matrices,  localized Wannier functions, and the quantum Hall effect. {\it J.   Math. Phys.} {\bf 51}, 015214 (2010)
	
	\bibitem[HS]{HeSj89} \textsc{Helffer, B.;  Sj\"ostrand, J.} :
        \'Equation de Schr\"odinger avec champ magnetique et equation de Harper. 
         In: {\it Schr\"odinger operators}, Lecture Notes in Physics 345,  Springer, Berlin, 1989, 118--197. 
         
        \bibitem[H\"o]{Hormander}
        \textsc{H\"ormander, L}: The Weyl calculus of pseudodifferential operators. 
        {\it Comm. Pure Appl. Math.} {\bf 32}, 359--443 (1979)
	
	\bibitem[IZH]{Irsigler et al 2019}
	\textsc{Irsigler, B., Zheng, J., Hofstetter, W.}: 
	Microscopic characteristics and tomography scheme of the local Chern marker. 
	{\it Phys. Rev. A} {\bf 100}, 23610 (2019)
	
	\bibitem[Ka]{Kato}
	\textsc{Kato, T.} : {\it Perturbation Theory for Linear Operators}. Springer, Berlin (1966)

         \bibitem[KS]{KellendonkSchulzBaldes19}
	\textsc{Kellendonk, J., Schulz-Baldes, H.}: 
	Quantization of edge currents for continuous magnetic operators. 
	{\it J. Funct. Anal.} {\bf 209}, 388--413 (2004)
	
	\bibitem[KSV]{KSV} \textsc{King-Smith,  R.~D. ; Vanderbilt, D.} : Theory of polarization of crystalline solids.
         {\it Phys. Rev. B} {\bf  47}, 1651--1654 (1993)

	\bibitem[Ki]{Kivelson1982}
	\textsc{Kivelson, S.}: Wannier functions in one-dimensional disordered systems: application
	to fractionally charged solitons.
	{\it Phys. Rev. B} \textbf{26}, 4269 (1982)
	
	 \bibitem[Ko]{Kohn59} \textsc{Kohn, W.} :  Analytic Properties of Bloch waves and Wannier Functions.
	{\it Phys. Rev.\ }{\bf 115}, 809 (1959)
	
	\bibitem[KO]{KohnOnffroy1973}
	\textsc{Kohn, W., Onffroy, J.~R.}: Wannier functions in a simple nonperiodic system. 
	{\it Phys. Rev. B} \textbf{8}, 2485 (1973)
	
	\bibitem[Kuc$_1$]{Kuchment2009}
	\textsc{Kuchment, P.} : Tight frames of exponentially decaying Wannier functions. {\it J. Phys. A: Math. Theor.} {\bf 42}, 025203 (2009).
	
	\bibitem[Kuc$_2$]{Kuchment2016}
	\textsc{Kuchment, P.} : An overview of periodic elliptic operators. {\it Bull. AMS} {\bf 53}, 343--414 (2016)
	
	\bibitem[Kun]{Kunz}
	\textsc{Kunz, H.}: The quantum hall effect for electrons in a random potential. {\it Commun. Math. Phys.} {\bf 112}, 121–145 (1987).
	
	
	\bibitem[LSi]{LeinfelderSimader1981} 
	\textsc{Leinfelder, H., Simader, C.G.} :   Schr\"odingers operators with singular magnetic vector potentials. 
	{\it Mathematische Zeitschrift} {\bf 176}, 1--19  (1981)
	
	\bibitem[LSt$_1$]{LuStubbs21I}
	\textsc{Lu, J, Stubbs, K.} : Algebraic localization implies exponential localization in non-periodic insulators. 
	Preprint \textsl{arXiv: 2101.02626} (2021)
	
	\bibitem[LSt$_2$]{LuStubbs21II}
	\textsc{Lu, J, Stubbs, K.} : Algebraic localization of Wannier functions implies Chern triviality in non-periodic insulators. 
	Preprint \textsl{arXiv:2107.10699} (2021)

	\bibitem[LT]{LudewigThiang19}
	\textsc{Ludewig. L.; Thiang, G.C.}: Good Wannier bases in Hilbert modules associated to topological insulators. {\it J. Math. Phys.} {\bf 61}, 061902 (2020)
	
	\bibitem[Ma]{Marcelli2022}
        \textsc{Marcelli, G.}: 
        Improved energy estimates for a class of time-dependent perturbed Hamiltonians. 
        {\it Lett. Math. Phys.} {\bf 112}, 51 (2022)

	\bibitem[MMMP]{MarcelliMonacoMoscolariPanati2018}
	\textsc{Marcelli, G.; Monaco, D.; Moscolari, M.; Panati, G.}: The Haldane model and its localization dichotomy. Rend. Mat. Appl. {\bf 39}, 307--327 (2018). Extended version available at \href{https://arxiv.org/abs/1909.03298}{\textsl{arXiv:1909.03298}}
	
	\bibitem[MM]{MarcelliMonaco2021}
	\textsc{Marcelli, G.; Monaco, D.}: Purely linear response of the quantum Hall current to space-adiabatic perturbations. 
	Accepted in {\it Lett. Math. Phys.}, preprint \textsl{arXiv:2112.03071} (2021)
	
	\bibitem[MPTa]{MarcelliPanatiTauber}
	\textsc{Marcelli, G.; Panati, G.; Tauber, C.}: Spin conductance and spin conductivity in topological insulators: Analysis of Kubo-like terms. {\it Ann. Henri Poincar\'{e}} {\bf 20}, 2071--2099 (2019)
		
	\bibitem[MPTe]{MarcelliPanatiTeufel}
	\textsc{Marcelli, G.; Panati, G.; Teufel, S.}:  A new approach to transport coefficients in the quantum spin Hall effect. {\it Ann. Henri Poincar\'e} {\bf 22}, 1069--1111 (2021).

        \bibitem[MYSV]{WannierReview} \textsc{Marzari, N.;  Mostofi A.A.; Yates J.R.; Souza I.; Vanderbilt D.} : 
        Maximally localized Wannier functions: Theory and applications. 
        \emph{Rev. Mod. Phys.\ }\textbf{84}, 1419 (2012)
        
        
        \bibitem[MV]{MaVa} \textsc{Marzari, N. ; Vanderbilt, D.} :
        Maximally localized generalized Wannier functions for composite energy bands. 
        {\it  Phys. Rev. B }{\bf 56}, 12847--12865 (1997)
        

	\bibitem[MP]{MonacoPanati2015}
	\textsc{Monaco, D.; Panati, G.} : Symmetry and localization in periodic crystals: 
	triviality of Bloch bundles with a fermionic time-reversal symmetry. {\it Acta App. Math.} {\bf 137}, 185--203 (2015).
	
	\bibitem[MPPT]{MonacoPanatiPisanteTeufel2018}
	\textsc{Monaco, D., Panati, G., Pisante, A., Teufel, S.}: Optimal decay of {W}annier functions in {C}hern and quantum
	 {H}all insulators.
	{ \it Commun. Math. Phys.} \textbf{359}, 61 (2018)
	
	
	\bibitem[MT]{MonacoTeufel}
	\textsc{Monaco, D., Teufel, S.}: Adiabatic currents for interacting electrons on a lattice. {\it Rev. Math. Phys.} {\bf 31}, 1950009 (2019).
	
	\bibitem[Mo]{Moscolari2018}
	\textsc{Moscolari, M.} :  {\it On the Localization Dichotomy for Gapped Quantum Systems.} 
	Ph.D. Thesis, \virg{La Sapienza} University of Rome (2018)
	
	\bibitem[MoPa]{MoscolariPanati}
	\textsc{Moscolari, M., Panati, G.} 
	Ultra Generalized Wannier bases: are they relevant for topological transport?. In preparation (2022)
	
	\bibitem[NB]{NakamuraBellissard}
		\textsc{Nakamura, S., Bellissard, J.}: Low energy bands do not contribute to quantum Hall effect. {\it Commun. Math. Phys.} {\bf 131} (2) 283 -- 305 (1990)
	
	\bibitem[Ne$_1$]{Nenciu1983} \textsc{Nenciu, G.} : Existence of the exponentially localised Wannier functions. 
	{\it Commun. Math. Phys.\,} {\bf 91}, 81--85 (1983).
	
	\bibitem[Ne$_2$]{Nenciu1991} \textsc{Nenciu, G.} : 
	Dynamics of band electrons in electric and magnetic fields: Rigorous justification of the effective Hamiltonians. 
	{\it Rev. Mod. Phys.\ }{\bf 63}, 91--127 (1991).
	
	\bibitem[NN]{NenciuNenciu1982}
	\textsc{Nenciu, A.,  Nenciu, G.} :   Dynamics of Bloch electrons in external electric fields. II. The existence of 
	Stark-Wannier ladder resonances. {\it  J. Phys.\  A} {\bf 15},  3313--3328  (1982). 
	
	\bibitem[NN$_1$]{NenciuNenciu1993} 
	\textsc{Nenciu, A.,  Nenciu, G.} :   Existence of exponentially localized Wannier functions for nonperiodic systems. 
	{\it Phys. Rev. B} {\bf 47}, 10112--10115  (1993)
	
	\bibitem[NN$_2$]{NenciuNenciu1998} 
	\textsc{Nenciu, A.,  Nenciu, G.} :   The existence of generalised Wannier functions for one-dimensional systems. 
	{\it Commun. Math. Phys.} {\bf 190}, 541--548  (1998)
	
	\bibitem[Pa]{Panati2007}
	\textsc{Panati, G.}: Triviality of Bloch and Bloch-Dirac bundles. {\it Ann. Henri Poincar\'e} {\bf 8},  995--1011 (2007)
	
	\bibitem[PP]{PanatiPisante2013}
	\textsc{Panati, G., Pisante, A.}: Bloch bundles, Marzari-Vanderbilt functional and maximally localized Wannier functions.
	 {\it Commun.  Math. Phys.\ } {\bf 322}, 835--875  (2013)
	 
	\bibitem[PST]{PanatiSparberTeufel2009}
	\textsc{Panati, G.; Sparber, C.; Teufel, S.} : Geometric currents in piezoelectricity. 
	\emph{Arch. Rat. Mech. Anal.\ }\textbf{91}, 387--422 (2009)
	
	\bibitem[Pr]{Prodan}
	\textsc{Prodan, E.}: On the generalized Wannier functions. \emph{J. Math. Phys.} \textbf{56}, 113511 (2015)
	

	\bibitem[ST]{SchulzBaldesTeufel13}
	\textsc{Schulz-Baldes, H.; Teufel, S.}: Orbital polarization and magnetization for independent particles in disordered media.
	 {\it Commun. Math. Phys.}  {\bf 319}, 649 (2013)
	
	\bibitem[Si$_1$]{Simon1982}
	\textsc{Simon, B.}: {S}chr\"odinger semigroups.
	{ \it Bulletin of the American Mathematical Society} \textbf{7}, 447
	(1982)
	
	\bibitem[Si$_2$]{Simon1983}
	\textsc{Simon, B.}: Holonomy, the quantum adiabatic theorem, and Berry's Phase. {\it Phys. Rev. Lett.} {\bf 51}, 2167 (1983)
	
	\bibitem[SWL$_1$]{StubbsWatsonLu20}
	\textsc{Stubbs, K. D., Watson A. B., Lu, J.}:
	Existence and computation of generalized Wannier functions for non-periodic systems in two dimensions and higher. 
	{\it Arch. Rat. Mech. Analysis }{\bf 243}, 1269--1323 (2022)
	
	\bibitem[SWL$_2$]{StubbsWatsonLu20II}
	\textsc{Stubbs, K. D., Watson A. B., Lu, J.}:
	 The iterated projected position algorithm for constructing exponentially localized generalized Wannier functions for periodic and nonperiodic insulators in two dimensions and higher. {\it Phys. Rev. B} {\bf 103}, 075125 (2021)
	
	
	\bibitem[Te]{Teufel}
	\textsc{Teufel, S.}: Non-equilibrium almost-stationary states and linear response for gapped quantum systems. {\it Commun. Math. Phys.} {\bf 373}, 621–653 (2020)
	
	\bibitem[Th]{Thouless1984}
	\textsc{Thouless, D.J.} : Wannier functions for magnetic sub-bands. {\it J. Phys. C} {\bf 17}, L325--L327 (1984)
	
	\bibitem[TKNN]{TKNN82}
	\textsc{Thouless, D.J.; Kohmoto, M.; Nightingale, M.P.; den Nijs, M.} : 
	Quantized Hall conductance in a two-dimensional periodic potential. {\it Phys. Rev. Lett.} {\bf 49}, 405--408 (1982)
	
	\bibitem[WL]{ELu2011}
	\textsc{Weinan, E.;  Lu, J.} :  
	The electronic structure of smoothly deformed crystals: {Wannier} functions and the Cauchy-Born rule. 
	{\it Arch. Ration. Mech. Anal.}\ {\bf 199}, 407--433 (2011)
	
\end{thebibliography}
\end{document}